\documentclass[lettersize,journal]{IEEEtran}
\usepackage{amsmath,amsfonts}
\usepackage{amsthm}
\usepackage{algorithmic}
\usepackage{algorithm}
\usepackage{array}
\usepackage[caption=false,font=normalsize,labelfont=sf,textfont=sf]{subfig}
\usepackage{textcomp}
\usepackage{stfloats}
\usepackage{url}
\usepackage{verbatim}
\usepackage{graphicx}
\usepackage{cite}
\usepackage{acronym}
\usepackage[usenames]{color}
\usepackage{cleveref}
\usepackage{amsmath,amssymb,amsfonts}
\usepackage{mathtools}
\crefname{figure}{Fig.}{Figs.}%
\usepackage{multirow}
\usepackage{makecell}
\usepackage{amsmath,amssymb,amsfonts}
\usepackage{mathtools}
\usepackage{svg}
\def\BState{\State\hskip-\ALG@thistlm}
\makeatother

\DeclareMathOperator{\Exp}{\mathbb{E}}

\DeclareMathAlphabet{\mathbit}{OML}{cmr}{bx}{it}
\newcommand{\B}[1]{\mathbf{#1}}

\usepackage{tikz}
\usetikzlibrary{calc}

\newcommand{\positiontextbox}[4][]{%
	\begin{tikzpicture}[remember picture,overlay]
		\node[inner sep=3pt, fill=yellow,align=left,draw,line width=1pt,#1] at ($(current page.north west) + (#2,-#3)$) {\parbox{.95\paperwidth}{#4}};
	\end{tikzpicture}%
}

\def\B{\mathbf}

\acrodef{OPTA}{optimum performance theoretically attainable}
\acrodef{5G}{fifth generation}
\acrodef{4G}{fourth generation}
\acrodef{MU}{multiuser}
\acrodef{eMBB}{Enhanced  Mobile  Broadband}
\acrodef{LTE}{Long Term Evolution}
\acrodef{EM}{electromagnetic}
\acrodef{OPTA}{Optimum Performance Theoretically
Attainable}
\acrodef{WSN}{wireless sensor network}
\acrodef{URLLC}{ultra-reliable  and  low latency  communications}
\acrodef{mMTC}{Massive Machine-Type Communications}
\acrodef{IoT}{Cellular Internet of Things}
\acrodef{V2X}{Vehicle-to-Everything}
\acrodef{MIMO}{multiple-input multiple-output}
\acrodef{MISO}{multiple-input single-output}
\acrodef{SISO}{single-input single-output}
\acrodef{B5G}{beyond 5G}
\acrodef{IRS}{intelligent reflecting surface}
\acrodef{LoS}{line-of-sight}
\acrodef{MU-MIMO}{multi-sensor MIMO}
\acrodef{EE}{energy efficiency}
\acrodef{MS}{multisensor}
\acrodef{SE}{spectrum efficiency}
\acrodef{LoS}{line-of-sight}
\acrodef{LBG}{Linde-Buzo-Gray}
\acrodef{NLoS}{non line-of-sight}
\acrodef{BS}{base station}
\acrodef{MEC}{mobile edge computing}
\acrodef{IoT}{Internet of Things}
\acrodef{6G}{sixth generation}
\acrodef{SNR}{signal-to-noise ratio}
\acrodef{DL}{deep learning}
\acrodef{ANN}{artificial neural network}
\acrodef{DRL}{deep reinforcement learning}
\acrodef{RL}{reinforcement learning}
\acrodef{DDPG}{deep deterministic policy gradient}
\acrodef{DQL}{Deep Q-learning}
\acrodef{PG}{policy gradient} 
\acrodef{SU-MISO}{single-sensor MISO}
\acrodef{MU-MISO}{multi-sensor multiple-input single-output}
\acrodef{CSI}{channel state information}
\acrodef{RF}{radio frequency} 
\acrodef{MAP}{Maximum A Posteriori}
\acrodef{AWGN}{additive white Gaussian noise}
\acrodef{MMSE}{minimum mean square error}
\acrodef{DPG}{deterministic policy gradient}
\acrodef{ReLU}{rectified linear unit}
\acrodef{tanh}{hyperbolic tangent}
\acrodef{SS}{single-stream}
\acrodef{DoF}{degrees of freedom} 
\acrodef{UAV}{unmanned aerial vehicle}
\acrodef{MRT}{maximum ratio transmitter}
\acrodef{TD3PG}{twin-delayed DDPG}
\acrodef{SIMO}{single-input and multiple-output}
\acrodef{SDR}{signal-to-distortion ratio}
\acrodef{JSCC}{joint source-channel coding}
\acrodef{MAP}{maximum a posteriori}
\acrodef{DQLC}{distributed quantizer linear coding}
\acrodef{MAC}{multiple access channel}
\acrodef{pdf}{probability density function}
\acrodef{MSE}{mean squared error}

\newtheorem{theorem}{Theorem}[section]
\newtheorem{lemma}[theorem]{Lemma}

\begin{document}

	\onecolumn
	\begingroup
	
	\setlength\parindent{0pt}
	\fontsize{14}{14}\selectfont
	
	\vspace{1cm} 
	\textbf{This is an ACCEPTED VERSION of the following published document:}
	
	\vspace{1cm} 
	P. Suárez-Casal, O. Fresnedo, D. Pérez-Adán and L. Castedo,``Lattice-Based Analog Mappings for Low-Latency Wireless Sensor Networks'', \textit{IEEE Internet of Things Journal}, vol. 10, n.o 19, pp. 17137-17154, Oct 2023, doi: 10.1109/JIOT.2023.3273194
	
	\vspace{1cm} 
	Link to published version: https://doi.org/10.1109/JIOT.2023.3273194
	
	\vspace{3cm}
	
	\textbf{General rights:}
	
	\vspace{1cm} 
	\textcopyright 2023 IEEE. This version of the article has been accepted for publication, after peer review. {Personal use of this material is permitted. Permission from IEEE must be obtained for all other uses, in any current or future media, including reprinting/republishing this material for advertising or promotional purposes, creating new collective works, for resale or redistribution to servers or lists, or reuse of any copyrighted component of this work in other works.}
	\twocolumn
	\endgroup
	\clearpage

\title{Lattice-Based Analog Mappings for Low Latency Wireless Sensor Networks}

\author{ Pedro Suárez-Casal, Óscar Fresnedo,~\IEEEmembership{Member,~IEEE,} Darian Pérez-Adán,~\IEEEmembership{Student Member,~IEEE,} and Luis Castedo,~\IEEEmembership{Senior Member,~IEEE}
\thanks{Pedro Suárez-Casal,  Óscar Fresnedo, Darian Pérez-Adán and Luis Castedo are with the Department of Computer Engineering, University of A Coruña, CITIC, Spain, e-mail: \{pedro.scasal, oscar.fresnedo, d.adan, luis\}@udc.es.}

}
\markboth{Draft}%
{IEEE Internet of Things Journal}

\maketitle

\begin{abstract}
We consider the transmission of spatially correlated analog information in a \ac{WSN} through fading \ac{SIMO} \acp{MAC} with low latency requirements. A lattice-based analog \ac{JSCC} approach is considered where vectors of consecutive source symbols are encoded at each sensor using an $n$-dimensional lattice and then transmitted to a multi-antenna central node. We derive a \ac{MMSE} decoder that accounts for both the multi-dimensional structure of the encoding lattices and the spatial correlation. In addition, a sphere decoder is considered to simplify the required searches over the multi-dimensional lattices. Different lattice-based mappings are approached and the impact of their size and density on performance and latency is analyzed. Results show that, while meeting low-latency constraints, lattice-based analog \ac{JSCC} provides performance gains and higher reliability with respect to the state-of-the-art \ac{JSCC} schemes.

\end{abstract}


\acresetall

\positiontextbox{10.8cm}{26.9cm}{\footnotesize \textcopyright 2023 IEEE. This version of the article has been accepted for publication, after peer review. Personal use of this material is permitted. Permission from IEEE must be obtained for all other uses, in any current or future media, including reprinting/republishing this material for advertising or promotional purposes, creating new collective works, for resale or redistribution to servers or lists, or reuse of any copyrighted component of this work in other works. Published version:
	https://doi.org/10.1109/JIOT.2023.3273194}

\section{Introduction}
\IEEEPARstart{T}{he} transmission of correlated information over fading \ac{SIMO} \ac{MAC} is a relevant problem in wireless communications, which is helpful to model many practical situations in \acp{WSN}, \acp{UAV}, \ac{IoT}, etc. Some mission-critical applications such as driverless vehicles, drone-based deliveries, factory automation, or artificial intelligence-based personalized assistants require uninterrupted and robust data exchange, i.e., \ac{URLLC}. Hence, efficient transmission schemes with short-packet data have to be implemented to meet both the reliability and latency requirements of wireless communication networks.

Conventional digital communications perform separately the optimization of source and channel coding according to the separation principle \cite{shannon2001mathematical}, and they are particularly well suited for high data rate applications by using long-size codewords.
However, this strategy suffers from many practical issues due to its high complexity and significant delay, as well as the need to optimize the encoders for given channel conditions, which leads to the requirement of accurately tracking the wireless channel and an adaptive design in time-varying communications scenarios.
For this reason, traditional digital systems may not be the best solution for real-time applications or for wireless transmissions with severe constraints on the acceptable communication delay. The separation principle also leads to suboptimal solutions when transmitting correlated information over \acp{MAC} \cite{5437416}. 

In this sense, \ac{JSCC} is an alternative approach where the source and the channel encoding are performed jointly in a single step. While digital approaches to \ac{JSCC} have been considered in the literature \cite{1292031,4350226,zhu2009distributed}, we instead focus in this work on analog \ac{JSCC} since sources are discrete-time continuous-amplitude symbols in most applications related to \acp{WSN} or \ac{IoT} systems. In addition, analog \ac{JSCC} is well suited for low-latency \ac{IoT} and \acp{WSN} communications (see e.g., \cite{9728732}) due to its capacity to achieve high transmission rates with very low complexity and almost zero delay. Analog \ac{JSCC} techniques are mostly focused on transforming the continuous-amplitude source symbols directly into channel symbols by using some analog mapping based on geometric curves \cite{shannon1949communication}. These mappings have been shown to closely approach the optimal system performance when considering the compression of non-correlated sources in \ac{AWGN} channels \cite{chung2000construction,hekland2009}, fading channels \cite{oscar_letter} and multiuser schemes \cite{Fresnedo15}. 

In this work, we address the design of analog \ac{JSCC} techniques for a \ac{WSN} \ac{SIMO} \ac{MAC} scenario, where non-cooperative sensors transmit their encoded source symbols to a centralized multi-antenna receiving node over fading channels with low-latency and high-reliability requirements \cite{8819994}. References \cite{8641107,9728732} are representative works on efficient mappings to accomplish the stringent requirements of \ac{URLLC} in some \ac{IoT} applications, e.g., control plants in Industry 4.0, factory automation scenarios, etc. %
In \cite{8641107}, the authors propose a short block length digital mapping to improve the decoding performance while preserving the low latency. The authors in \cite{9728732} propose a novel analog \ac{JSCC} mapping which is well suited for wirelessly powered sensor nodes in \ac{IoT}.

Concerning analog mappings, different zero-delay \ac{JSCC} schemes (i.e., codewords of size $n=1$) have been investigated in the literature for different versions of the considered scenario. In \cite{floor15}, a zero-delay analog \ac{JSCC} mapping was proposed to transmit multivariate Gaussian sources over an \ac{AWGN} \ac{MAC}. The resulting mapping combined the use of a scalar quantizer and linear transmission which can be seen as a practical zero-delay approach to the optimal mapping for such a scenario (assuming infinite block size)\cite{lapidoth10}.  
On the other hand, modulo-like mappings were proposed for the orthogonal transmission of correlated sources in a \ac{SISO} \ac{MAC} \cite{Wernersson09,Mehmetoglu15}, and for \ac{SIMO} \ac{MAC} systems with enough degrees of freedom to exploit the source correlation at reception \cite{Suarez17}. 

In general, most works on analog \ac{JSCC} focus on zero-delay mappings due to the difficulty of designing and optimizing mappings of larger dimensions. This restriction significantly limits the practical utility of analog \ac{JSCC} with respect to traditional digital approaches which can consider different encoding block sizes. In addition, these zero-delay mappings are able to provide satisfactory performance in terms of transmission reliability, but they still remain relatively separated from the theoretical optimum performance for the considered WSN scenario. Therefore, a systematic strategy to design low-latency analog JSCC mappings for codewords of arbitrary size $n$ is fundamental to extend the utilization of analog JSCC to a large range of communication scenarios and satisfy the strict requirements of URLLC applications on transmission reliability. Moreover, a flexible implementation of analog \ac{JSCC} mappings for codeword sizes $n>1$ is timely to obtain higher system performance while keeping low latency requirements, especially due to the performance degradation of digital schemes when using very short codeword sizes \cite{polyanskiy2010channel}.

In this context, the work in \cite{5962670} represents a first attempt to construct analog \ac{JSCC} schemes with arbitrary codeword size on the basis of the lattice theory. %
The authors  consider zero-delay and non-zero delay mappings which use well-known lattices with different dimensions. 
Specifically, modulo-like mappings ($n=1$),  the $D_4$ lattice ($n=4$) \cite{conway2013sphere} and the $E_8$ lattice ($n=8$) \cite{conway2013sphere} are employed for the encoding of source symbols to be next transmitted over a fading channel. As commented in \cite{5962670}, $D_4$ and $E_8$ are the densest packing lattices for $n=4$ and $n=8$, respectively,  and hence they are suitable candidates to obtain a satisfactory system performance. %
However, the parameter optimization is based on an exhaustive search which is infeasible for large lattice dimensions. In addition, the decoding operation is performed by a three-step procedure that first estimates the symbols of the uncoded sensors and next employs such estimates as side information for the remaining sensors. Hence, this strategy does not allow to jointly exploit the source correlation. Furthermore, \cite{5962670} only considers lattice dimensions up to dimension $n=8$. 

In any way, the results in \cite{5962670} provide the intuition that the use of $n$-dimensional lattices for the design of analog \ac{JSCC} mappings is a promising strategy. A comprehensive analysis for the construction of ``good" lattices in different dimensions can be found in \cite{conway2013sphere}, particularly for the sphere packing problem. 
The list of the known densest lattices for different dimensions can also be looked up in \cite{Nebe_url}. An interesting example is the Leech lattice \cite{leech1967notes}, which is the unique densest sphere packing lattice for $n=24$ \cite{cohn2009optimality}. Unfortunately, the computational cost of encoding when using the densest lattice-based mappings exponentially increases with the codeword size $n$ because it requires to find the closest point in the $n$-dimensional lattice space \cite{1019833}.  %
Hence, alternative lattice constructions must be considered to balance the system performance and the encoding complexity as the codeword size $n$ becomes larger. 

From this perspective, an attractive type of lattices are the so-called Craig's lattices \cite{craig1978extreme}.
The mechanism to construct these lattices allows to adjust their minimal norm and thus their density for any arbitrary dimension $n = p-1$, where $p$ is a prime number \cite{conway2013sphere}. Therefore, a feasible approach for the analog encoding of the source symbols is to use Craig's lattices with a suitable density for the analog mapping. This approach reduces the computational effort required to find the closest lattice points during the encoding operation. Hence, the use of these lattices can enable the design of practical analog \ac{JSCC} schemes that achieve better reliability with reasonable block sizes for low latency applications.

Leveraging all these previous ideas, we address in this work the transmission of spatially correlated discrete-time analog sources in a \ac{WSN} by means of multi-dimensional lattice-based analog \ac{JSCC} mappings. First, blocks of $n > 1$ source symbols (or measurements) are encoded at each sensor node with an analog mapping constructed from a suitable lattice of dimension $n$ by considering the low latency requirements, i.e., we focus on small block sizes $(n)$. The resulting codewords of size $n$ are next transmitted to a centralized receiver over a fading \ac{SIMO} \ac{MAC}. At this central node, the estimates of the transmitted symbols are jointly decoded taking into account the codeword size and the spatial correlation of the source symbols. Therefore, the main contributions of this work can be summarized as:

\begin{itemize}
\item {A lattice-based analog \ac{JSCC} system is designed and optimized for the transmission of blocks of symbols with codeword size $n$ that offers system performance gains while preserving low latency requirements. The proposed design hence allows analog \ac{JSCC} techniques to be a practical alternative to conventional digital schemes for \ac{URLLC} systems. This design is sufficiently flexible, both in terms of the parameters optimization and the decoding procedure, to work with different codeword sizes and efficiently exploit the source spatial correlation.

}
\item {Craig's lattices are considered to reduce the computational cost of the analog \ac{JSCC} encoding operation for the largest codeword sizes. The possibility of adjusting the lattice density allows us to balance the trade-off between the system performance and the computational complexity. In addition, we propose an alternative construction of Craig's lattices based on using vectors of minimum norm and exploiting their polynomial nature.
}
\item A performance evaluation by means of computer simulations showing the advantages of the proposed system design and the use of Craig's lattices. In particular, performance gains are determined for scenarios with non-orthogonal configurations or moderate correlations where zero-delay mappings exhibit lower performance and thus, low reliability. In addition, the impact of the block size and the lattice density on the system performance is analyzed.
\end{itemize}
\subsection{Organization}
The remainder of this paper is structured as follows. In \Cref{Latt}, we present a brief review of some preliminary concepts corresponding to the lattice theory. In \Cref{SM}, the considered \ac{SIMO} \ac{MAC} system model for \acp{WSN} is detailed. The design of the lattice-based analog \ac{JSCC} scheme is addressed in \Cref{LBE}, where different lattice constructions and their main parameters are explained. In addition, the derivation of the optimal \ac{MMSE} estimation combined with a sphere decoder to produce the symbol estimates is also described in this section. The computer experiments to evaluate the system performance and the obtained results are discussed in \Cref{SR}. Finally, \Cref{Con} is devoted to the conclusions.
\vspace*{-0.2cm}
\subsection{Notation}
The following notation is employed: $a$ is a scalar and $\B{a}$ is a vector, $[\B{{A}}]_{i,j}$ is the entry on the $i$-th row and the $j$-th column of the matrix $\B{A}$. Transpose and conjugate transpose of $\B{A}$ are $\B{A}^{T}$ and $\B{A}^{H}$, respectively. $\parallel \mathbf{A} \parallel$ represents the $2$-norm of $\mathbf{A}$. The operations tr($\cdot$), diag$(\cdot)$,  $\lfloor \cdot \rceil$ and $\lfloor \cdot \rfloor$ stand for the trace of a matrix, the diagonal matrix with the argument in the main diagonal, the element-wise rounding and the floor operation, respectively. 
The operator $\mid \cdot \mid$ represents the absolute value for a scalar argument, the matrix determinant in case of a matrix argument, and cardinality in case of set arguments. $\Re (\cdot)$ represents the real part of a complex-valued argument.
Finally, the expectation is denoted by $\Exp [\cdot ]$ and $\otimes$ represents the Kronecker product. Table \ref{Tsym10}
summarizes the notation employed throughout this paper.
\vspace*{-0.2cm}
\begin{table}[h!]
\centering
\caption{{Notation.}}\label{Tsym10}
	\setlength{\tabcolsep}{5pt}
	\def\arraystretch{1.3}
\begin{tabular}{m{2.3cm}|m{5.8cm}}
\hline
 {\textbf{Symbol / Operator}} & {\textbf{Description}}  \\ \hline
$\left(\mathbf{\bullet }\right)^{T}$, $\left(\mathbf{\bullet}\right)^{H}$ & Transpose, conjugate transpose\\
${\parallel \mathbf{\bullet} \parallel}$, ${\parallel\mathbf{\bullet} \parallel}_{F}$ &  2-norm, Frobenius norm\\ 
$\Re(\bullet)$, $\Im(\bullet)$ &Real part, imaginary part\\ 
$\mathbb{R}$, $\mathbb{C}$ & Set of real numbers, set of complex numbers\\
$\mathbf{I}_N$,$\mathbf{0}_N$ & Identity matrix with size $N,$ all zeros matrix with size $N$\\
$\left[\mathbf{{\mathbf{A}}}\right]_{i,j}$&  Entry on the $i$-th row and the $j$-th column of $\mathbf{A}$\\
$\left[\mathbf{\mathbf{A}}\right]_{i,:}$, $\left[\mathbf{\mathbf{A}}\right]_{:,j}$    &$i$-th row of $\mathbf{A}$, $j$-th column of $\mathbf{A}$\\
tr$\left(\mathbf{A}\right)$, diag$\left(\bullet\right)$& Trace of $\mathbf{A}$, diagonal matrix with the argument in the main diagonal\\
$\lfloor \bullet \rceil$, $\lfloor \bullet \rfloor$, $\left\lceil\bullet \right\rceil$ & Element-wise rounding, floor operation, ceiling  operation\\
$\mid a \mid$, $\mid \mathbf{A} \mid$, $\mid \mathcal{A} \mid$ & Absolute value of $a$, determinant of the matrix $\mathbf{A}$, cardinally of the set $\mathcal{A}$\\
mod$\;(\mathbf{a},\textbf{b}$) & Element-wise modulo operation that returns for each vector element\\
$\mathcal{N}_{\mathbb{C}} (\boldsymbol{\mu},\B{C})$ & 
Circularly-symmetric complex normal distribution with mean $\boldsymbol{\mu}$ and covariance matrix $\mathbf{C}$\\
 $\otimes$&   Kronecker product\\
$\mathbb{E} [\bullet ]$& Statistical expectation \\
\hline
\end{tabular}
\end{table}

\section{Fundamentals of Lattices}\label{Latt}
This section introduces the theoretical fundamentals of lattices that will be used throughout this work. %
An $n$-dimensional lattice ${\varLambda}$ is defined as a discrete set of vectors in $\mathbb{R}^n$ which form a group under vector addition. These vectors will be referred to as the lattice points. 
A lattice ${\varLambda}$ will be defined by its generator matrix $\mathbf{M}=[\boldsymbol{\nu}_1,\ldots, \boldsymbol{\nu}_n]\in \mathbb{R}^{n \times n}$, where the column vectors $\boldsymbol{\nu}_j, \;j=1,\ldots, n$ are the basis vectors. Therefore, the $i$-th lattice point will be generated as 
${\mathbf{x}_{\text{l}}}_i=\mathbf{M}\mathbf{l}_i$,  $\forall i=1,\ldots, M$, where %
$\mathbf{l}_i\in\mathbb{Z}^n$ is a vector of integers and $M$ is assumed to be large enough. %
Lattice points are usually the representation points of a surrounding region having a particular shape (cubic, hexagonal, sphere, ...) in the $n$-dimensional space. A shape that allows to tile the entire $n$-dimensional space with the aggregation of the surrounding regions of all the lattice points is termed a fundamental region of the lattice.  
The Voronoi region is the fundamental region which contains all points in an $n$-dimensional euclidean space closer to its representation lattice point than to any other lattice point.
The different possibilities of partitioning or covering an $n$-dimensional euclidean space with a lattice lead to different well-known problems in algebra such as the sphere packing, the covering packing or the quantization problem.

\subsection{Sphere Packing}
Sphere packing seeks to fill the $n$-dimensional space with non-overlapping equal-sized spheres. %
Unlike cubic regions, there is always some wasted space when packing spheres.
The minimization of such wasted space is  still an unsolved problem which can be stated as follows: determine the largest number of balls with the same radius $r$ that can be packed into a largely empty $n$-dimensional region. As observed in \Cref{fig:sphere_packing}, the wasted space between spheres, also known as deep holes, corresponds to the points whose minimum distance to any point in the lattice $r_t$ is larger than the radius of the spheres, i.e., $r_t > r$. 

In the sphere packing problem, the lattice points $\varLambda = \{{\mathbf{x}_{\text{l}}}_1,{\mathbf{x}_{\text{l}}}_2,\ldots,{\mathbf{x}_{\text{l}}}_M\}$ correspond to the central points of the spheres. The lattice density is defined as the ratio between the space that is occupied by the spheres and the total volume, i.e., 
\begin{equation}\label{den}
    \Delta = \frac{V_{{1}} p_{\text{r}}^n} {V({\varLambda})},
\end{equation}
where $V_{1}$ is the volume of an $n$-dimensional sphere of radius $r=1$, $V({\varLambda})$ is the volume of the lattice ${\varLambda}$, which is determined as \cite{conway2013sphere}
\begin{equation}
    V({\varLambda})= \text{det}\left(\mathbf{M}\mathbf{M}^T\right)^{1/2},
    \label{eq:volume}
\end{equation}
and $p_\text{r}$ is the packing radius defined as half of the minimal distance between lattice points, i.e., $p_\text{r}=\frac{1}{2} \sqrt{\mu({\varLambda})}$, where 
$
    \mu({\varLambda}) = \text{min}_{i\neq j}\lbrace \parallel {\mathbf{x}_{\text{l}}}_i - {\mathbf{x}_{\text{l}}}_j\parallel^2 \rbrace, ~\forall ~{\mathbf{x}_{\text{l}}}_i,{\mathbf{x}_{\text{l}}}_j \in {\varLambda}
$
is the minimum norm of the lattice.
Another important parameter is the center density which is defined as 
\begin{equation}\label{cd}\psi=\frac{\Delta }{V_{\text{1}}} = \frac{p_{\text{r}}^n} {V({\varLambda})} .\end{equation}
Note that the center density provides a more intuitive idea of how dense a packing lattice is. For instance, in packings with unitary-radius spheres, the center density directly indicates the number of centers (lattice points) per unit volume.
\begin{figure}[htpb]
	\centering
	\includegraphics[width=0.6\columnwidth]{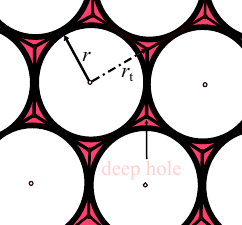}
	\caption{Circle packing problem and deep holes.}
	\label{fig:sphere_packing}\vspace{-0.4cm}
\end{figure}
 
\subsection{Lattice-Based Quantization}
In general, quantization consists in partitioning an $n$-dimensional space into $M$ non-overlapping regions each of them represented by a representation point usually termed centroid that will be interpreted as a point in a lattice. Typically, quantization regions are Voronoi regions. This way, each point in the source space is quantized to the closest point in the lattice and the quantization error is minimized.

An $n$-dimensional lattice-based quantizer comprises a lattice $\varLambda = \{{\mathbf{x}_{\text{l}}}_1,{\mathbf{x}_{\text{l}}}_2,\ldots,{\mathbf{x}_{\text{l}}}_M\} \subset \mathbb{R}^n$, defined via a generator matrix $\mathbf{M}$, and a quantization function ${Q}_{{\varLambda}}(\cdot)$ which maps any input vector $\mathbf{s} \in \mathbb{R}^n$ into the closest lattice point ${\mathbf{x}_{\text{l}}}_i \in {\varLambda}$. The quantization region associated to the $i$-th lattice point ${\mathbf{x}_{\text{l}}}_i=\mathbf{M}\mathbf{l}_i$, $\forall i=1,\ldots,M$, will be its Voronoi region defined as
\begin{equation}
    \Omega_{\varLambda}\;({\mathbf{x}_{\text{l}}}_i)=\lbrace \mathbf{s}: \Vert \mathbf{s}- {\mathbf{x}_{\text{l}}}_i \Vert \leq  \Vert \mathbf{s}- {\mathbf{x}_{\text{l}}}_j \Vert \rbrace,\; \forall i\neq j.
\end{equation}
Therefore, the quantization function is mathematically defined as 
${Q}_{{\varLambda}}(\mathbf{s}) = {\mathbf{x}_{\text{l}}}_i,  \forall \;\mathbf{s} \in \Omega_{\varLambda}\;({\mathbf{x}_{\text{l}}}_i).$

A possible metric to measure the quantization error, by considering $M$ to be a very large number, is the average \ac{MSE}, i.e.,  
\begin{equation}
    \varepsilon = \frac{1}{n} ~\sum_{i=1}^M \; \int\displaylimits_{\Omega_{\varLambda}({\mathbf{x}_{\text{l}}}_i)} \hspace*{-0.2cm}\Vert \mathbf{s} - {\mathbf{x}_{\text{l}}}_i \Vert^2~ p(\mathbf{s}) ~d\mathbf{s},
    \label{eq:quant_error}
\end{equation} 
where the factor $1/n$ is introduced for fair comparison among quantizers of different dimensions.

Although quantization and sphere packing are different lattice design problems, there is an intrinsic relation between them. A ``good'' packing lattice implies efficiently covering the region of interest with non-overlapping spheres, minimizing the deep holes, and thus minimizing the probability of having points very distant from the lattice points. Therefore, this feature apparently leads to the minimization of the quantization error in \eqref{eq:quant_error} if those ``good" packing lattices were used to solve the quantization problem. In this sense, optimal lattice-based quantizers are only known for low dimensions while sphere packing is a widely studied problem in multi-dimensional lattice theory. Indeed, the densest sphere packing lattices %
have been shown to provide satisfactory performance when used for quantization \cite{conway2013sphere}.

On the other hand, quantization with extremely dense lattices leads to a huge computational effort when considering relatively large dimensions ($n$ values above 16).
There exist several algorithms in the literature to alleviate this problem but in any case, their computational cost exponentially increases with the dimension $n$ and the lattice density. Therefore, there is a trade-off between quantization error minimization and computational complexity which can be balanced through the lattice density factor defined in \eqref{den}.

\section{System Model}\label{SM}

 \begin{figure}[t!]
	\centering
	\includegraphics[width=0.99\columnwidth]{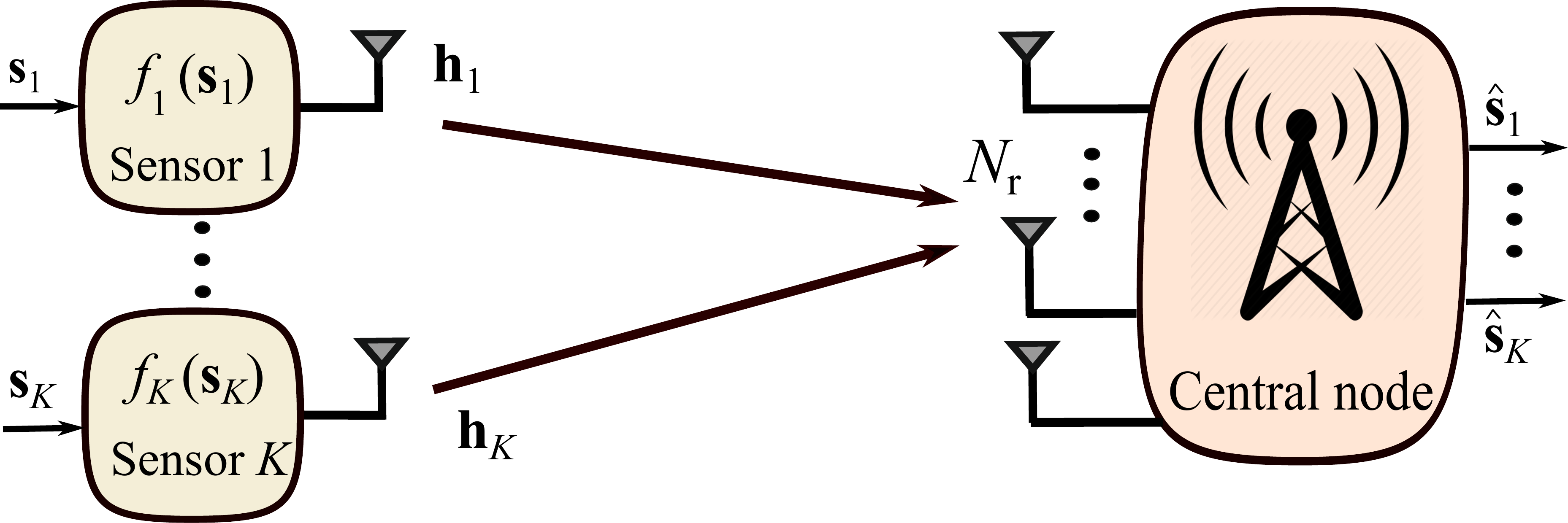}
	\caption{{Block diagram of the considered analog }\ac{JSCC}  {communication system.}}
	\label{fig:block_diagram}
\end{figure}
 Let us consider the uplink of a communication system
 {as shown in {%
\cref{fig:block_diagram}}\hspace{0pt}%
, }where $K$ single-antenna nodes transmit their source information over a fading \ac{MAC} to a %
central node with $N_{\text{r}}$ antennas. This communication system can be used to model practical scenarios of \acp{WSN} and \ac{IoT} systems.
Henceforth, we will refer to this model as a $K\times N_\text{r}$ \ac{SIMO} \ac{MAC} system. 
In this scenario, the complex-valued analog source symbol transmitted by the $k$-th {sensor } at the discrete-time instant $t$ is denoted by $s_{k,t}$. The source symbols transmitted by the $K$ {sensors } at the time instant $t$ {are } represented by the vector $\B{s}_t = [s_{1,t}, \cdots ,s_{K,t}]^T \in \mathbb{C}^K$, which is modeled as a multivariate circularly symmetric complex-valued zero-mean Gaussian distribution with covariance matrix $\B{C}_{\B{s}}= \mathbb{E}\left[\B{s}_t\B{s}_t^H\right]$. The elements $[\B{C}_{\B{s}}]_{i,j} = \rho_{i,j}$ represent the spatial correlation between the $i$-th and $j$-th source symbols of $\B{s}_t$. Without loss of generality, we assume that $\rho_{i,i} = \sigma^2, \forall i$.  We also consider that the source symbols at different time instants $t$ are statistically independent so that we only consider the  spatial correlation between the sensors. Such spatial correlation is assumed to not depend on the time instant, i.e., remains constant for a sufficiently long period of time. The \ac{pdf} of $\B{s}_t$ is therefore given by
\begin{equation}
{p_{\mathbf{s}_t}}({\B{s}}) = \frac{1}{\pi^K |\B{C}_{\B{s}}|}\exp\left(-\B{s}^H \B{C}_{{\B{s}}}^{-1} {\B{s}}\right).
\label{eq:pdf_source}
\end{equation}

A block of source symbols is individually encoded at each sensor prior to its transmission by means of an analog mapping function. We employ lattice-based mappings  %
which transform the continuous-amplitude source symbols into the complex-valued encoded symbols to be transmitted. In particular, a vector of $n/2$ consecutive complex-valued source symbols $\B{s}_k$ corresponding to the  $k$-th sensor is mapped with the function $f_k(\cdot): \mathbb{C}^{\frac{n}{2}} \rightarrow \mathbb{C}^{\frac{n}{2}}$ to produce the encoded vector
$
\B{x}_k = f_k(\B{s}_k), \forall  k=1, \ldots, K, 
$
with  $\B{s}_k = [s_{k,1}, s_{k,2}, \ldots, s_{k,\frac{n}{2}}]^T$ and $\B{x}_k = [x_{k,1}, x_{k,2}, \ldots, x_{k,\frac{n}{2}}]^T$ . The encoded vectors are next transmitted to the common central node over the fading \ac{MAC} by using $n/2$ channel uses to produce the received signal vectors $\mathbf{y}_t \in \mathbb{C}^{N_{\text{r}}}$ as
\begin{equation}
\mathbf{y}_t = \sum_{k=1}^K \B{h}_k x_{k,t}+ \B{n}_t,\quad\forall t=1,\ldots,n/2, 
\label{eq:simo_mac_model}
\end{equation}
where $\B{h}_k\in \mathbb{C}^{N_{\text{r}}}$ is the channel response from the $k$-th sensor to the central node, and $\B{n}_t = [n_{1,t}, \ldots, n_{N_{\text{r}},t}]^T \sim \mathcal{N}_{\mathbb{C}}(\B{0},\sigma_{\text{n}}^2\B{I})$ is the \ac{AWGN} component. %
The channel responses are assumed to remain the same at least during the transmission of a block of $n/2$ symbols.
 Each sensor is also subject to an individual power constraint such that $\mathbb{E}[|x_{k,t}|^2] \le {P_{\text{T}}}_k, \forall k=1,\ldots,K$. %

 \vspace{0.1cm}
 The expression in \eqref{eq:simo_mac_model} can be rewritten in a more compact way as 
\begin{align}
\B{y}_t = \B{H}\B{x}_t +\B{n}_t, \quad~\forall t=1,\ldots,n/2,
\label{eq:MAC_model_t}
\end{align}
where $\B{H}\in \mathbb{C}^{N_\text{r}\times K}$ stacks all the  channel responses, i.e., $\B{H} = [\B{h}_1, \cdots, \B{h}_K]$, and $\B{x}_t = [x_{1,t}, \ldots, x_{K,t}]^T $ contains all sensors encoded symbols at a given time instant $t$. %

\vspace{0.1cm}
At the central node, an estimate of the transmitted sensor symbols $\hat{\mathbf{s}}_t$ is obtained from the \ac{MAC} signal $\B{y}_t$ by using an appropriate decoding function that jointly decodes the symbols received during the corresponding $n/2$ channel uses. Since we are considering the analog encoding and transmission of the source information, the sensor information will always be recovered with a certain level of distortion which is measured in terms of the \ac{MSE} between the source and the estimated symbols, i.e., $\xi = \mathbb{E}[||\hat{\B{s}}_t - \B{s}_t||^2]$. In this case, the \ac{MMSE} estimator constitutes the optimal decoding strategy.
\subsection{Real-Valued Equivalent Model}\label{realimag}
In the considered system model, the variables corresponding to the source symbols, channel matrices and \ac{AWGN} components are complex-valued with uncorrelated real and imaginary parts (circularly symmetric). However, analog lattice-based mappings work in the real domain as the lattices are defined as a group in $\mathbb{R}^n$. Following the same approach as in \cite{Suarez17},
the complex-valued system model will be transformed into an equivalent real-valued one. For this, the real and imaginary parts of the different variables are separated and stacked into a unique vector with twice the size, whereas the equivalent real-valued channel matrix is rearranged as
$
\tilde{\B{H}} =  \Re(\B{H})\otimes  \B{I}_2  + \Im(\B{H})\otimes \B{E},
$
with  $\B{E} = [0 -1; 1 ~0]$. In addition, the source and noise covariance matrices must be adapted such that $\tilde{{\B{s}}}\sim \mathcal{N}_{\mathbb{R}}(\B{0}, \B{C}_{{\tilde{\text{s}}}} )$, with $\B{C}_{\tilde{\text{s}}} = \B{C}_{\B{s}}\otimes \frac{1}{2}\B{I}_2$, and $\tilde{\mathbf{n}}\sim \mathcal{N}_{\mathbb{R}}(\B{0},\frac{\sigma_{\text{n}}^2}{2}\B{I})$.
\section{Multi-Dimensional Lattice-Based Analog JSCC}\label{LBE}
Along this work, we will assume that sensors in the $K\times N_\text{r}$ \ac{SIMO} \ac{MAC} system under consideration use analog \ac{JSCC} mappings to encode their source symbols individually. As shown in \cite{Suarez17}, modulo-like mappings provide satisfactory performance when the block size is $n=1$ (i.e., zero-delay). The symbols encoded with modulo-like mappings are the difference between the source symbols and the central point of their corresponding interval. In the following, this idea is extended to the consideration of an arbitrary dimension $n > 1$

Let us consider the real-valued equivalent model in \Cref{realimag}. The parametric definition of the multi-dimensional lattice-based mapping functions is stated as
\begin{equation}
\tilde{\B{x}}_k = {f}_k(\tilde{\B{s}}_k) = \delta_k\left(\tilde{\B{s}}_k - {Q}_{{\varLambda}}(\tilde{\B{s}}_k)\right) = \delta_k(\tilde{\B{s}}_k-{\alpha}_k\B{M}\B{l}_k),
\label{eq:sensor_mapping}
\end{equation}
where $\tilde{\B{s}}_k \in \mathbb{R}^{n}$ and $\tilde{\B{x}}_k\in \mathbb{R}^{n}$ comprise the real and the imaginary parts of the $k$-th sensor source and encoded symbols, respectively. 
The operator $Q_{\varLambda}(\cdot)$ determines the lattice point (centroid) closest to its argument. Note that this quantization step depends on the considered lattice $\varLambda$ which is generated by the matrix $\B{M}$ and scaled with the parameter $\alpha_k$.
The mapping parameters $\delta_k$ and $\alpha_k$ are used to fulfill the power constraint at the sensors and to adjust the distance between any two lattice points, respectively, whereas $\B{l}_k\in \mathbb{Z}^n$ is the index vector which represents the coordinates of the specific Voronoi region where the vector of source symbols falls into during the encoding process (cf. \cite{9097147}). It is worth remarking that the same lattice ${\varLambda}$ is employed for all the $K$ sensors but with different scaling factors $\{\alpha_k,\delta_k\}$, $k=1, \ldots, K$. 

The encoded symbols are obtained by determining the difference vector between the source symbols and their corresponding centroid. In this sense, the system performance improves when the difference vectors have the smallest possible norm as long as the correct decoding of the source symbols is guaranteed. For given power constraints, this fact results in the use of larger power factors at each sensor, $\delta_k$, and thus minimizes the resulting decoding distortion. Hence, an adequate optimization of the mapping parameters $\{\alpha_k,\delta_k\}$ is essential to ensure the lattice-based system works properly.

The lattice-based mappings for all the $K$ sensor symbols can be rewritten in a compact way as
\begin{equation}
\tilde{    \mathbf{x}}_{\text{c}}=f(\tilde{\mathbf{s}}_{\text{c}})=\mathbf{D}(\tilde{\mathbf{s}}_{\text{c}}-\mathbf{B}\mathbf{l}),
\label{eq:encoded_symbols_compact}
\end{equation}
where the vector $\tilde{\mathbf{s}}_{\text{c}}\in \mathbb{R}^{nK}$ stacks the real and imaginary parts of the blocks of $n/2$ complex-valued symbols for all the $K$ sensors, i.e., $\tilde{\mathbf{s}}_{\text{c}} = [\tilde{\B{s}}_1^T, ~\tilde{\B{s}}_2^T, \ldots, ~\tilde{\B{s}}_K^T ]^T$,
$\mathbf{D}=\text{diag}(\delta_1,\ldots,\delta_K)\otimes\mathbf{I}_{n}$, whereas $\mathbf{B}=\mathbf{U}\otimes\mathbf{M}$ with $\mathbf{U}=\text{diag}(\alpha_1,\ldots,\alpha_K)$.  Note that the vector $\tilde{\mathbf{s}}_{\text{c}}$ comprises all the source symbols for the $K$ sensors and for the $n/2$ channel uses. In turn, the vector $\mathbf{l}=\left[\mathbf{l}_1^T,\ldots,\mathbf{l}_K^T\right]^T$ such that $\mathbf{l}\in\mathbb{Z}^{nK}$ stacks all the coordinates corresponding to the $K$ Voronoi regions which the $K$ vectors of source symbols fall into.

By considering the discrete nature of the lattices, the compact expression for the mapping function $f(\cdot)$ can be rewritten from an equivalent piece-wise formulation given by
     \begin{equation}\label{fpw}
       f(\mathbf{\tilde{\mathbf{s}}_{\text{c}}})=
             f_i(\tilde{\mathbf{s}}_{\text{c}})  %
            ~~\text{if} ~\tilde{\B{s}}_{{c}} \in \Omega_{\varLambda}(\B{l}_{i}),  
    \end{equation}
where
\begin{equation}
    f_i(\tilde{\mathbf{s}}_{\text{c}}) = \mathbf{D}(\tilde{\mathbf{s}}_{\text{c}}-\mathbf{B}\mathbf{l}_i).
\end{equation}
Note that the vector $\mathbf{l}_i \in \mathbb{Z}^{nK}$ identifies a specific combination of $K$ Voronoi regions, one for each sensor, which is denoted as $\Omega_{\varLambda}(\B{l}_{i})$. %
In particular, the function $f_i(\cdot)$ is defined for the combination of regions corresponding to the $K$ Voronoi regions, $\Omega_{\varLambda}(\B{l}_{k,i})$, where each sensor symbol vector $\tilde{\B{s}}_k$ falls into. Recall that $\B{l}_{k,i}$ stands for the coordinates vector corresponding to the lattice point closest to $\tilde{\B{s}}_k$ and $\Omega_{\varLambda}(\B{l}_{k,i})$ represents the corresponding Voronoi region at the $k$-th sensor. At this point, it is also worth remarking that for each feasible combination of $K$ Voronoi regions, we will have a particular function $f_i(\cdot)$ with the corresponding vector $\mathbf{l}_i$. For convenience, we will define the set $\mathcal{L}$ containing all the integer-valued vectors of dimension $nK$ which identify a feasible combination of $K$ Voronoi regions given the source distribution, i.e., $\mathbf{l}_i\in\mathcal{L}$.

The blocks of $n/2$ complex-valued encoded symbols at each {sensor } are transmitted over the \ac{MAC} by using $n/2$ channel uses. Accordingly, decoding is applied to the entire block of $n/2$ complex-valued received symbols. For that reason, we extend the compact formulation for the one-shot \ac{MAC} signal in \eqref{eq:MAC_model_t} to include all the   symbols transmitted by all the sensors during $n/2$ consecutive channel uses, and also considering the real-valued equivalent model. Hence, the compact representation of the received symbols is
\begin{align}
\tilde{\B{y}}_{\text{c}} = \tilde{\B{H}}_{\text{c}}\tilde{\B{x}}_{\text{c}} +\tilde{\B{n}}_{\text{c}},
\label{eq:y_compact}
\end{align}
where $\tilde{\B{H}}_{\text{c}}=\tilde{\B{H}}\otimes \mathbf{I}_{\frac{n}{2}}$ such that $\tilde{\B{H}}_{\text{c}}\in \mathbb{R}^{{n} N_{\text{r}}\times {n}K}$, $\tilde{\B{x}}_{\text{c}}=\left[\tilde{\B{x}}^T_1,\ldots, \tilde{\B{x}}^T_K\right]^T$ with $\tilde{\B{x}}_{\text{c}}\in \mathbb{R}^{nK}$ stacking all the $K$ blocks of $n/2$ encoded symbols, and $\tilde{\B{n}}_{\text{c}}=\left[\tilde{\B{n}}^T_1,\ldots, \tilde{\B{n}}^T_n\right]^T$ is the noise affecting the received symbols during the $n/2$ channel uses. Recall that $\tilde{\B{x}}_{\text{c}}$ is obtained according to \eqref{eq:encoded_symbols_compact}. The vector of received symbols, $\tilde{\B{y}}_{\text{c}} \in \mathbb{R}^{n\times N_r}$, is employed to produce the estimates of the $K$  blocks of transmitted symbols by using the \ac{MMSE}-based procedure to be explained in \Cref{MMSED}.

As mentioned, the encoding operation requires to find the closest lattice point for each vector of sensor symbols $\tilde{\B{s}}_k$. We employ a refined version of the Pohst's algorithm \cite{1019833}, which has been shown to be faster than other known methods like, e.g., Kannan's algorithm \cite{kannan1983improved} or the conventional Pohst's algorithm \cite{pohst1981computation}. This iterative algorithm searches for the optimal lattice point inside a hypersphere in $\mathbb{R}^n$ which should contain such a point. %
The search implies exploring all the lattice points which fall into the considered $n$-dimensional hypersphere to determine the one with minimum Euclidean distance. Therefore, their computational complexity not only grows with the lattice dimension, but also with the lattice density since the number of lattice points inside the hypersphere will be much larger. In practice, the closest point algorithms are able to efficiently deal with the densest packing lattices up to $n\approx 24$, whereas they exhibit an impractical complexity for larger dimensions. This issue will be circumvented by using Craig's lattices whose density can be properly adjusted for a given dimension. This fact allows us to increase the encoding lattice dimension at the expense of reducing the lattice density. Hence, the use of Craig's lattices can contribute to improve the system performance in terms of transmission reliability with minimum impact on the communication delay as we are still using reasonable small block sizes. 

In the following subsection, we introduce the fundamentals of Craig's lattices and present an alternative construction to reduce the computational complexity of Craig's lattice-based encoding.    
\subsection{Craig's Lattices}\label{CsL}
Craig's lattices are constructed from the ring of integers in a cyclotomic field. Let $\zeta_p$ be a primitive $p$-th root of unity being $p$ an odd prime. The elements of the ring of integers $\mathbb{Z}[\zeta_p]$ in the cyclotomic field $\mathbb{Q}[\zeta_p]$ are represented as
\begin{equation}
    \omega= P(\zeta_p) = a\zeta_p^{n-1}+,\ldots,+b\zeta_p^{}+c,
\end{equation}
where the polynomial coefficients are restricted to be integer values. In this case, $n = p-1$ determines the order of the elements in the ring and the dimension of the resulting lattices.

An $n$-dimensional Craig's lattice, usually denoted as ${A}^{(m)}_{n}$, is generated from the elements of the ideal $(1-\zeta_p)^m$, with $m$ a positive integer, in the cyclotomic ring of integers $\mathbb{Z}[\zeta_p]$  \cite{craig1978cyclotomic}. Thus, an $n$-dimensional Craig's lattice is given by the subset of polynomials in the ring of integers $\mathbb{Z}[\zeta_p]$ which are multiples of $(1-\zeta_p)^m$. 
Alternatively, the lattice points can be obtained as the vector representation of the elements of the ideal generated by $(1-x)^m$ in the quotient ring $\mathbb{Z}[x]/(x^p-1)$. Such elements can be obtained as
\begin{equation}
    P(x)(1-x)^m = Q(x)(x^p-1)+R(x),
    \label{eq:polynomial_def}
\end{equation}
for some polynomial $P(x) \in \mathbb{Z}[x]$. With this formulation, the elements of the ideal generated by $(1-x)^m$ would be the remainders $R(x)$, and the corresponding lattice points in ${A}^{(m)}_{n}$ will be the vectors that contain the coefficients of the polynomials $R(x)$. %
However, for the encoding operation, we need to obtain the generator matrix for the Craig's lattice $A^{(m)}_{n}$. In the following, we present the standard construction of this matrix and propose an alternative one based on using minimum norm vectors which reduces the computational cost of the encoding operation.

\vspace{0.0cm}
\subsubsection{Standard Construction}
The standard construction for the generator matrix of a Craig's lattice is based on the idea of constructing sub-lattices from the well-known ${A}_n$ lattice. 
In that case, given an $n$-dimensional lattice ${\varLambda}$, the difference lattice ${\boldsymbol{\Delta}}^T{\varLambda}$ satisfies ${\boldsymbol{\Delta}}^T{\varLambda} \subseteq {\varLambda}$, where
\begin{align*}\label{str}
\def\arraystretch{0.8}
\arraycolsep=1.0pt
{\boldsymbol{\Delta}}^T = \left[\begin{array}{c c c c c c }
\;1 & -1 & \;\;0 & \cdots & 0 & \;\;0 \\
\;0 & \;\;1 & -1 & \cdots & 0 & \;\;0 \\
\;\cdot & \; \cdot & \;\;\cdot &\cdots & \cdot & \;\;\cdot \\
\;0 & \;0 & \;\;0 & \cdots & 1 & -1 \\
\;1 & \;\;1 & \;\;1 & \cdots & 1 & \;\;\;2\\
\end{array}\right]
\end{align*}
 is a $n \times n$ matrix. Thereby, the generator matrix of a Craig's lattice ${A}^{(m)}_{n}$ can be defined as
\begin{equation}
    \B{M}_{n}^{(m)}=\boldsymbol{\Delta}^{m-1}\B{M}_{n}, ~~~~\forall ~m \leq n/2,
\end{equation}
where $\B{M}_{n}$ is the generator matrix for the ${A}_n$ lattice.
In practice, this procedure is equivalent to considering the following polynomial sequence 
    $\{ P_1(x) = 1,
    P_2(x) = x,
    P_3(x) =x^2,  
    \ldots, 
    P_{n}(x)=x^{n-1}\}$
in \eqref{eq:polynomial_def} to generate the $n$ basis vectors $\{\boldsymbol{\nu}_1,\ldots, \boldsymbol{\nu}_n\}$ corresponding to the $n$ columns of the generator matrix of the Craig's lattice. 
Because of the polynomial properties, this construction is in turn equivalent to setting the first column of $ \B{M}_{n}^{(m)}$ to $\boldsymbol{\nu}_1 = P_1(x)(1-x)^m$ and the remaining columns to cyclic shifts of this primary vector 
\cite[Ch. 8, Th. 10]{conway2013sphere}. However, note that the norm of $\boldsymbol{\nu}_1$ increases as $m$ becomes larger and this negatively impacts the encoding operation. 
Several works have shown that a desirable property of the generator matrices to optimize the computational effort of closest point algorithms is the fact that the scalar product of their columns is as small as possible \cite{babai1986lovasz,qiao2008integer}. Considering the particular structure of the generator matrix of a Craig's lattice, this is equivalent to obtaining a primary vector with the minimum norm. 
Inspired by this idea, we propose in the following an alternative and equivalent construction for the generator matrix of Craig's lattices based on using cyclic shifts of a primary basis vector $\boldsymbol{\nu}_1$ with the minimum norm for a given $m$.

\vspace{0.0cm}
\subsubsection{Alternative Construction}\label{sub:alter_const}

According to the lattice theory \cite{conway2013sphere}, the minimum norm of a Craig's lattice ${A}^{(m)}_{n}$ is at least $2m$. This result implies that we can always find some lattice vector with norm $2m$.
It can be observed that polynomials of the form
\begin{equation}
    T(\zeta_p) = \prod_{j\in \mathcal{J}(m)}(\zeta_p^j-1), 
\end{equation}
where $\mathcal{J}(m)$ is the subset of exponent indexes, which depends on $m$, and have the ability to produce the desired minimum norm vectors for practical values of $m$. Therefore, the elements of the Craig's lattice can be obtained by the ideal generated by $T(\zeta_p)$ instead of the standard ideal given by $(1-\zeta_p)^m$. In this case, the generator matrix with minimum norm columns can be constructed by cyclically shifting the primary vector $\boldsymbol{\nu}_1$, obtained from $T(\zeta_p)$ with \eqref{eq:polynomial_def}.

For the previous procedure to result in the same Craig's lattice as the standard construction, the ideal generated by $T(\zeta_p)$ should be equivalent to the one generated by $(1-\zeta_p)^m$. This condition is guaranteed by the following Lemma. 
\begin{lemma}
Given the cyclotomic ring of integers $\mathbb{Z}[\zeta_p]$ with $p$ prime, the ideal $G=(T_1(\zeta_p) \cdot \ldots \cdot T_m(\zeta_p))\subset \mathbb{Z}[\zeta_p]$, generated by the product of $m$ polynomials in the form $T_i(\zeta_p)=\zeta_p^{k_i}-1, k_i\ge 1$, is equal to the ideal $I=((1-\zeta_p)^m)$.
\label{lemma:equal_ideals}
\end{lemma}

\begin{proof}
Since we can make the decomposition $T_i(\zeta_p) = \zeta_p^{k_i}-1 =(\zeta_p-1)\times(\zeta_p^{k_i-1}+\ldots+1)$, for any element $x\in G$ we have that $x=Q(\zeta_p)\times T_1(\zeta_p)\times\ldots \times T_m(\zeta_p)=Q(\zeta_p)\times(\zeta_p-1)^m \times(\zeta_p^{k_1-1}+\ldots+1)\times\ldots\times(\zeta_p^{k_m-1}+\ldots+1)\in I$, and hence $G\subset I$.

In addition, any element of $\mathbb{Z}[\zeta_p]$ with the form $z_{k}(\zeta_p) = (\zeta_p^k-1)/(\zeta_p-1)=\zeta_p^{k-1}+\ldots+1$ has a multiplicative inverse, and therefore for any element $x\in I$ we have that $x=Q(\zeta_p)\times(\zeta_p-1)^m=Q(\zeta_p)\times(\zeta_p-1)^m \times z_{k_1}(\zeta_p)\times z_{k_1}^{-1}(\zeta_p)\times \ldots \times z_{k_m}(\zeta_p)z_{k_m}^{-1}(\zeta_p)=Q(\zeta_p)\times T_1(\zeta_p) \times \ldots \times T_m(\zeta_p)\times z_{k_1}^{-1}(\zeta_p)\times\ldots \times z_{k_m}^{-1}(\zeta_p)\in H$. Hence $I\subset G$ and $I=G$.
\end{proof}
In practice, we have observed that the alternative construction provides the same performance as the standard construction with a significantly lower computational cost which enables the use of Craig's lattices for higher dimensions and larger values of $m$.

\vspace{0.0cm}
\subsubsection{Craig's Lattice Parameters}
In this subsection, we present important parameters related to the Craig's lattices. The determinant of the Craig's lattice ${A}^{(m)}_{n}$ is 
$
    \operatorname{det}\left({A}^{(m)}_{n}\right) = (n+1)^{2m-1},
$
where $n = p-1$, with $p$ an odd prime, and $m<n/2$ \cite[Chapter 8]{conway2013sphere}. The volume can be computed from \eqref{eq:volume} and is hence given by
\begin{equation}
     V_{\text{c}} = V\left({A}^{(m)}_{n}\right) = (n+1)^{(m-1)/2}.
    \label{eq:vol_Craig}
\end{equation}
The minimum norm for the Craig's lattice ${A}^{(m)}_{n}$ is at least $2m$, i.e., $
    \mu\left({A}^{(m)}_{n}\right)\geq 2m.
$
Therefore, the choice of $m$ directly impacts on the lattice packing radius which is given by
\begin{equation}
{p_{\text{r}}}_{\text{c}}=\frac{1}{2}\sqrt{\mu \left({A}^{(m)}_{n}\right)} \geq \sqrt{\frac{m}{2}}.
\label{eq:packing_radius_Craig}
\end{equation}
Using \eqref{eq:vol_Craig} and \eqref{eq:packing_radius_Craig}, a lower bound for the density of the Craig's lattice is given by
\begin{equation}
    \Delta_{\text{c}} = \frac{V_{{1}} {p_{\text{r}}}_{\text{c}}^n} {V_{\text{c}}}.
    \label{eq:Craig_density}
\end{equation} 
As observed, we can obtain different lattice densities depending on the parameter $m$. The lower bound given by \eqref{eq:Craig_density} is maximized for the value $m = m_0$ where
\begin{equation}
    m_0 = \left\lfloor \frac{1}{2}\;\frac{n}{\operatorname{log}_e(n+1)}\right\rceil \label{eq:m0}.
\end{equation}
This is the value of $m$ for which an $n$-dimensional Craig's lattice achieves its maximum density \cite[Chapter 8]{conway2013sphere}.
Finally, the center density (number of lattice points per volume unit) is  
\begin{equation}
\psi_{\text{c}}=\frac{\Delta_{\text{c}} }{V_{\text{1}}} = \frac{{p_{\text{r}}}_{\text{c}}^n} {V_{\text{c}}}.
\end{equation}

\Cref{tablerp} shows the base 2 logarithm of the center density for Craig's lattices having different values of $n$ and $m$. For the dimension $n=16$, the highest value for the center density is obtained when $m=3$, whereas for the dimensions $n=36$, $n=52$ and $n=60$ the densest Craig's lattices are obtained for $m=5$, $m=7$ and $m=7$, respectively. These values agree with the integers resulting from the formula in \eqref{eq:m0}. Recall that the computational cost of the closest point algorithms depends on the number of lattice points in the search region, and the parameter $\psi_{\text{c}}$ is an intuitive indicator of this number.

{{%
On the other hand, the quantization error for a given lattice is an interesting metric to predict the performance of the resulting lattice-based mappings in the considered scenario. The quantization error for some particular lattices can be determined analytically but in most cases (like for Craig's lattices), it needs to be computed numerically according to expression in \eqref{eq:quant_error}.  
\Cref{tableq} }\hspace{0pt}%
compares the quantization errors obtained for different lattice constructions. First, we have considered Craig's lattices with dimensions $n \in\{ 16, 36, 52\}$. For the dimensions $n=16$ and $n=36$, we have analyzed the densest lattice construction, i.e., the parameter $m$ is set to $m=3$ and $m=5$, respectively (see \Cref{tablerp}). For $n=52$, we have considered $m=3$, since other denser lattices for such a dimension lead to impractical computational complexity.

In addition, lattices based on the construction $A$
\cite[Chapter 4]{conway2013sphere} }
are also analyzed in \Cref{tableq}, since it  constitutes  a relatively simple lattice implementation based on the use of digital error correction codes. For this reason, we considered this approach as an interesting benchmark to show the suitability of Craig's lattices for our problem. The lattices obtained with the construction $A$ are represented as $\text{AC}[\mathcal{H}(r,n)]$, where $\mathcal{H}(r,n)$ refers to the employed digital code with source block size $r$ and codeword size $n$. Given the required small codeword sizes, we have decided to consider binary cyclic codes such that their generator polynomial guarantee the largest minimum distance for the resulting code. In particular, we have considered cyclic codes with codeword sizes $n=16$ and $n=32$, and source block length $r=5$, and $r=6$, respectively. Finally, the quantization error achieved by the well-known Leech lattice ($n=24$) {%
\cite{cohn2009optimality,conway2013sphere}}\hspace{0pt}%
,  the Barnes-Wall $(\text{BW}_{16})$ lattice ($n=16)$ {%
\cite{conway2013sphere}}\hspace{0pt}%
, and the $E_8$ lattice ($n=8$) {%
\cite{conway2013sphere,5962670}}\hspace{0pt}%
  are also included in the comparison. While the quantization error provided by this latter lattice is well-known (cf. \cite[Table 2.3]{conway2013sphere}), the error for the rest of the lattices and those obtained with the construction $A$ is numerically computed.

{As shown in \Cref{tableq}, the $A^{(3)}_{52}$ Craig's lattice leads to the lowest quantization error, whereas the $\text{AC}[\mathcal{H}(5,16)]$ construction obtains the highest quantization error. In general, lattices based on the construction $A$ clearly provide the highest quantization errors followed by the $E_8$ and the $\text{BW}_{16}$ lattices. This behavior suggests that the lattices based on construction $A$ are a counter-productive choice for the design of analog mappings in this scenario.   
On the contrary, it is interesting to remark that the $A^{(3)}_{52}$ and the $A^{(4)}_{36}$ lattices provide a slightly lower
error than that assessed with the Leech lattice. Therefore, it is reasonable to consider  Craig's lattices as suitable candidates to implement lattice-based  }\ac{JSCC} schemes.
\begin{table*}[t!]
	\centering
		\vspace{0mm}
	\caption{Center density ($\psi_{\text{c}}$) versus $m$.}
	\vspace{0mm}
	\label{tablerp}
	\setlength{\tabcolsep}{5pt}
	\def\arraystretch{1.7}
\begin{tabular}{cccccccccc}
		\hline
		Lattice $\mathbf{size\; (n)}$&
	    \textbf{$m=2$}  &
		$m=3$ &
		$m=4$ &
		$m=5$  & 
		$m=6$ &  
		$m=7$ &
		$m=8$ &
	$m=9$  &
	$m=10$ \\

				16&
	     -6.13 &
	-5.54	 &
	-6.31	 &
	-7.82	  & 
	-9.80	 &  
	 -12.11 &
	-14.66 &
		-17.38 &
		 -20.26 \\
		 
		 36&
		 	-7.81&
	     -2.49 &
	-0.23	 &
	0.35	 &
	-0.12	  & 
	-1.33	 &  
	 -3.07 &
		-5.22 &
		 -7.70 \\

		 52&
	     -8.60 &
	0.89	 &
	5.95	 &
	8.60	  & 
9.71	 &  
	 9.76 &
		9.04 &
		 7.73 &
		 5.95 \\
		 
		 60&
	     -8.90&
	2.72	 &
	9.24	  & 
	12.97	 &  
	 14.93&
		15.67 &
		15.51 &
		14.69 &
		 13.32 \\

		\hline

	\end{tabular}
	\vspace{0mm}
\end{table*}
\begin{table*}
    \vspace*{-2mm}
	\centering
	\caption{Quantization errors with different lattice constructions.}
	\vspace{0mm}
	\label{tableq}
	\setlength{\tabcolsep}{5pt}
	\def\arraystretch{1.8}
\begin{tabular}{ccccccccc}
		\hline
\textbf{Lattice}& $E_8$&BW$_\mathbf{16}$&Leech&$A^{(3)}_{16}$&$A^{(4)}_{36}$&$A^{(3)}_{52}$&$\text{AC}[\mathcal{H}(5,16)]$&$\text{AC}[\mathcal{H}(6,32)]$\\
\textbf{Quantization error}&
 0.0710&
0.0682&
 0.0656&
 0.0688&
 0.0649&
 0.0643&
 0.0929&
 0.0835\\
		\hline
	\end{tabular}
	\vspace{0mm}
\end{table*}

\subsection{MMSE Decoding}\label{MMSED}
When considering the transmission of analog sources, \ac{MMSE} decoding is optimum as it minimizes the observed distortion. The \ac{MMSE} estimator of the source symbols $\tilde{\mathbf{s}}_{\text{c}}$ from the received symbols $\tilde{\mathbf{y}}_{\text{c}}$ is given by 
\begin{equation}\label{mmsei}
    \hat{\mathbf{s}}_{\text{c}}=\mathbb{E}[\tilde{\mathbf{s}}_{\text{c}} \mid \tilde{\mathbf{y}}_{\text{c}}] =\int \tilde{\mathbf{s}}_{\text{c}}\; p_{\text{s}}(\tilde{\mathbf{s}}_{\text{c}} \mid \tilde{\mathbf{y}}_{\text{c}}) ~d\tilde{\mathbf{s}}_{\text{c}}.
\end{equation}
    By employing the piece-wise definition of the mapping function in \eqref{fpw} and the compact expression for $\tilde{\mathbf{y}}_{\text{c}}$ in \eqref{eq:y_compact}, the conditional probability  $p_{\text{s}}(\tilde{\mathbf{s}}_{\text{c}} \mid \tilde{\mathbf{y}}_{\text{c}})$ can be expressed as $$p_{\text{s}}(\tilde{\mathbf{s}}_{\text{c}}  \mid \tilde{\mathbf{y}}_{\text{c}})\propto \sum_{i=1}^{|\mathcal{L}|}r_i(\tilde{\mathbf{y}}_{\text{c}},\tilde{\mathbf{s}}_{\text{c}}),$$ where
    \begin{equation}
        r_i(\tilde{\mathbf{y}}_{\text{c}},\tilde{\mathbf{s}}_{\text{c}})\propto
        \left\{ \begin{array}{lcc}
             \phi_i\; g(\tilde{\mathbf{s}}_{{c}}\mid \boldsymbol{\mu}_i, \Sigma)  & \text{if} ~\tilde{\B{s}}_{{c}} \in \Omega_{\varLambda}(\B{l}_{i})\\ 
             0 & \text{otherwise}, \\
             \end{array}
             \right.
    \label{eq:sum_truncated}         
    \end{equation}
    with
    \begin{align}
   \notag g(\tilde{\mathbf{s}}_{\text{c}}\mid \boldsymbol{\mu}_i, \boldsymbol{\Sigma})\hspace{-0mm}=&\hspace{-0mm}\Big(\big(2\pi\big)^{Kn}| \boldsymbol{\Sigma} | \Big)^{-1/2} \\&\times\text{exp}\left(\hspace{-0mm}-\frac{1}{2}(\tilde{\mathbf{s}}_{\text{c}} - \boldsymbol{\mu}_i)^T\boldsymbol{\Sigma}^{-1}(\tilde{\mathbf{s}}_{\text{c}}-\boldsymbol{\mu}_i)\right),
    \end{align}

    \begin{equation}
        \phi_i = \text{exp}\left(-\frac{1}{2}\left(\sigma_{\text{n}}^{-2} \parallel \tilde{\mathbf{y}}_{\text{c}}-\tilde{\mathbf{H}}_{\text{c}}\mathbf{D}\mathbf{B}\mathbf{l}_i\parallel^{2}-\boldsymbol{\mu}_i^T\boldsymbol{\Sigma}^{-1}\boldsymbol{\mu}_i \right)\right),
    \end{equation}
    
       \begin{equation}
        \boldsymbol{\mu}_i = \frac{1}{\sigma_{\text{n}}^{2}}\boldsymbol{\Sigma}\mathbf{D}^T\tilde{\mathbf{H}}_{\text{c}}^{T}(\tilde{\mathbf{y}}_{\text{c}}+\tilde{\mathbf{H}}_{\text{c}}\mathbf{D}\mathbf{B}\mathbf{l}_i),
    \end{equation}
     and
    \begin{equation}\label{pao}
        \boldsymbol{\Sigma}=\left( \frac{1}{\sigma_{\text{n}}^2}\mathbf{D}^T\tilde{\mathbf{H}}^{T}_{\text{c}}\tilde{\mathbf{H}}_{\text{c}}\mathbf{D}+\mathbf{C}_{\tilde{\text{s}}}^{-1} \right)^{-1}.
    \end{equation}

\vspace{0.0cm}    
The steps required to obtain this result are similar to those explained in \cite[Appendix A]{Suarez17} for \ac{MMSE} estimation using modulo-like functions but considering mapping functions from $n$-dimensional lattices.
Recall that the \ac{MMSE} integral in \eqref{mmsei} is decomposed into a sum of terms weighted by their corresponding factor $\phi_i$. An important remark is that the function $g(\tilde{\mathbf{s}}_{\text{c}}\mid\boldsymbol{\mu}_i,\boldsymbol{\Sigma})$ actually represents the \ac{pdf} of a truncated multivariate Gaussian with mean $\boldsymbol{\mu}_i$ and covariance matrix $\boldsymbol{\Sigma}$, which is restricted to the corresponding $Kn$-dimensional region given by the aggregate of the $K$ Voronoi regions $\Omega_{\varLambda}(\B{l}_{i})$.
Therefore, we will compute the \ac{MMSE} estimates of the sensor symbols as
\begin{equation}\label{suma_gauss}
    \hat{\mathbf{s}}_{\text{c}}=\frac{\sum_i\phi_i{\Theta}\left( \Omega_{\varLambda} (\B{l}_{i});\boldsymbol{\Sigma},\boldsymbol{\mu}_i\right)}{\sum_i\phi_i
    {\Phi}(\Omega_{\varLambda}(\B{l}_{i});\boldsymbol{\Sigma},\boldsymbol{\mu}_i)},
\end{equation}
where ${\Theta}( \Omega_{\varLambda} (\B{l}_{i});\boldsymbol{\Sigma},\boldsymbol{\mu}_i)=\int\displaylimits_{\Omega_{\varLambda}({\mathbf{l}_{i}})}\;\tilde{\mathbf{s}}_{\text{c}} ~g(\tilde{\mathbf{s}}_{\text{c}}\mid \boldsymbol{\mu}_i,\boldsymbol{\Sigma})~{d}\tilde{\mathbf{s}}_{\text{c}}$ is the mean of a  $Kn$-dimensional multivariate Gaussian truncated to the region given by $\Omega_{\varLambda} (\B{l}_{i})$, and  ${\Phi}( \Omega_{\varLambda}( \B{l}_{i});\boldsymbol{\Sigma},\boldsymbol{\mu}_i)=\int \displaylimits_{ \Omega_{\varLambda} (\B{l}_{i})}\;g(\tilde{\mathbf{s}}_{\text{c}}\mid \boldsymbol{\mu}_i,\boldsymbol{\Sigma})~{d}\tilde{\mathbf{s}}_{\text{c}}$ represents the cumulative distribution of a multivariate Gaussian variable in the aggregated region $\Omega_{\varLambda}( \B{l}_{i})$.

 \Cref{fig:Gaussianas} shows an illustrative example of a bi-dimensional space which is partitioned into four feasible regions $\Omega_{\varLambda}(\B{l}_{i})$ by using lattice-based mappings. For simplicity, we assume rectangular regions but, in the general case, the shape of the truncated regions is given by the corresponding Voronoi regions. The pdf of the truncated Gaussian functions at each region is represented by contour lines with different colors indicating different probabilities.
According to \eqref{eq:sum_truncated}, we have different Gaussian functions (with mean $\boldsymbol{\mu}_i$ and same covariance matrix $\boldsymbol{\Sigma}$) weighted by the factor $\phi_i$. As shown, the maximum value of these functions could fall outside the truncated region due to the channel and noise effects.
Note also that the size of the truncated regions is given by the parameters $\alpha_k$ within the matrix $\B{B}$.
\begin{figure}[htpb]
	\centering
	\includegraphics[width=0.9\columnwidth]{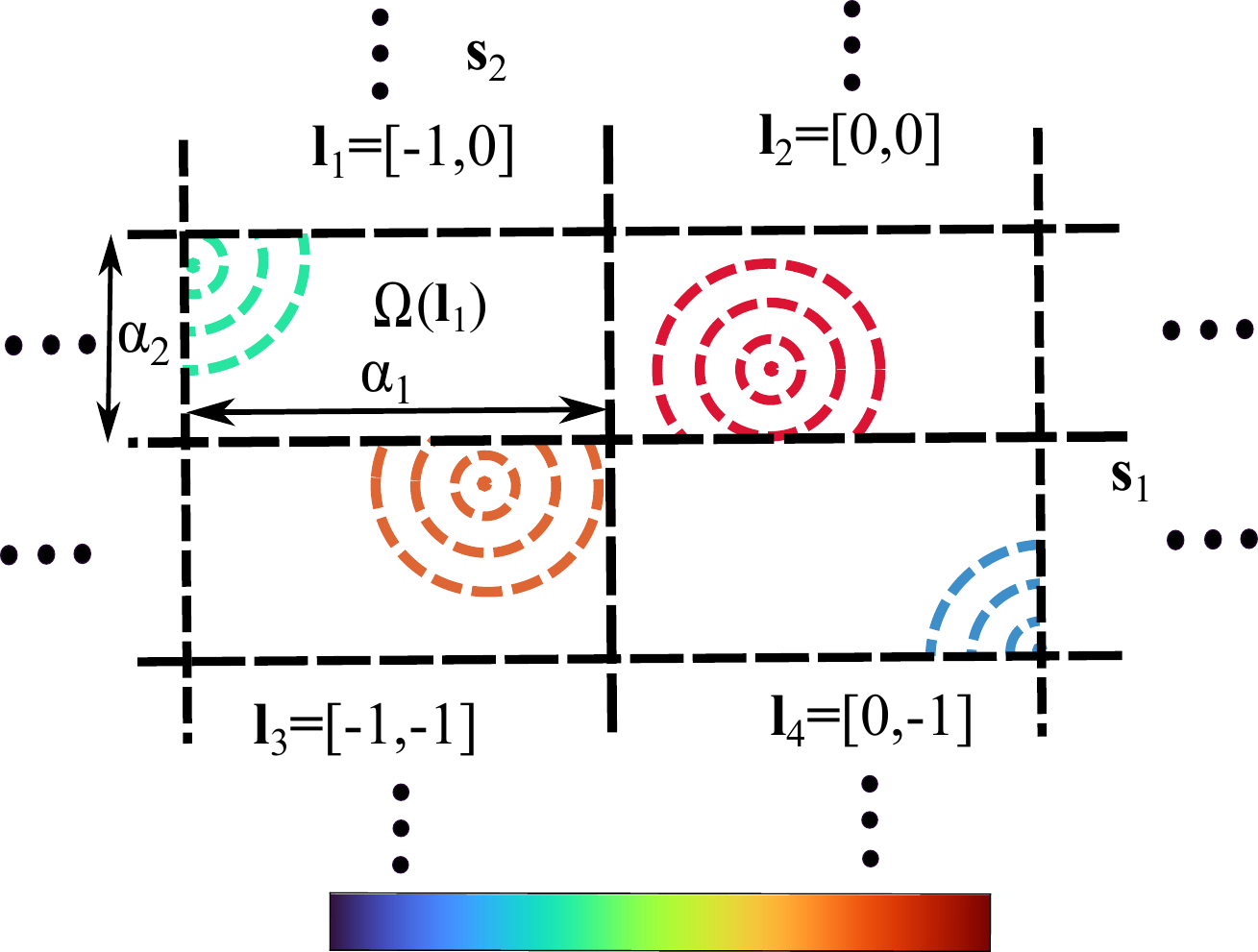}
	\caption{Example of the partition of the source space using bi-dimensional truncated Gaussian variables.}
	\label{fig:Gaussianas}
\end{figure}
\subsubsection{Sphere Decoding with MAP Estimation}
The computation of the \ac{MMSE} estimates with \eqref{suma_gauss} requires to determine the integrals of the truncated Gaussian functions over the $Kn$-dimensional regions defined by $\Omega_{\varLambda} (\B{l}_{i})$. This computation leads to two important problems: 
\begin{itemize}
    \item The number of potential combinations of Voronoi regions increases as the lattice dimension $n$ becomes larger, i.e., the cardinality of the set $\mathcal{L}$ dramatically grows with $n$.
    \item The integration of Gaussian functions over $Kn$-dimensional regions is an extremely difficult problem with an unaffordable computational complexity, even for small values of $K$ and $n$ and considering hyperspheres with radius $p_\text{r}$ as integration regions.    
\end{itemize}

The first problem can be alleviated by considering only those truncated regions $\Omega_{\varLambda}(\B{l}_{i})$ with the largest associated weights $\phi_i$. In such a case, the problem can be formulated as the search of the candidate vectors $\B{l}_i$ such that their corresponding weight factor $\phi_i$ exceeds a given threshold $T$, i.e.,
\begin{equation}
    \phi_i (\B{l}_i)=\text{exp}\left( -\frac{1}{2} \left(  \sigma_{\text{n}}^2 \parallel \tilde{\mathbf{y}}_{\text{c}}+\tilde{\mathbf{H}}_{\text{c}}\mathbf{D}\mathbf{B}\mathbf{l}_i\parallel^2-\boldsymbol{\mu}_{i}^T\boldsymbol{\Sigma} \boldsymbol{\mu}_{i} \right) \right)>T.
    \label{eq:phis}
\end{equation}
Hence, the set of relevant candidate vectors is constructed as
$
    \mathcal{L}_d = \{\B{l}_i \in \mathcal{Z}^{Kn}~|~ \phi_i (\B{l}_i) > T\}.  
$
This problem was already approached for the case of modulo-like mapping functions with Rayleigh channels in \cite{Suarez17}. The proposed solution is based on transforming the $Kn$-dimensional search space into a ``decoding" lattice $\varLambda_{\text{d}}$ whose points correspond to all the possible vectors $\B{l}_i$, and then using a sphere decoder to select those lattice points which fall inside a hypersphere with a particular radius. 

Following an approach similar to that explained in \cite[Appendix B]{Suarez17}, we obtain the following Gram matrix for the decoding lattice 
\begin{equation}\label{latr}
    \mathbf{A}_{\text{d}}=\frac{1}{2}\mathbf{B}^T\mathbf{D}^T\tilde{\mathbf{H}}_{\text{c}}^{T}(\sigma_{\text{n}}^2\mathbf{I}+\tilde{\mathbf{H}}_{\text{c}}\mathbf{D} \mathbf{C}_{\tilde{\text{s}}} \mathbf{D}^T  \tilde{\mathbf{H}}_{\text{c}}^T)^{-1}\tilde{\mathbf{H}}_{\text{c}}\mathbf{D}\mathbf{B},
\end{equation}
and the vector
$$
    \mathbf{l}_0=(\mathbf{B}^T\mathbf{D}^T\tilde{\mathbf{H}}_{\text{c}}^T\tilde{\mathbf{H}}_{\text{c}}\mathbf{D}\mathbf{B})^{-1}\mathbf{B}^T\mathbf{D}^T\tilde{\mathbf{H}}_{\text{c}}^T\tilde{\mathbf{y}}_{\text{c}}
$$
for the center of the sphere where the candidate vectors will be searched for. Therefore, the generator matrix for the decoding lattice $\varLambda_{\text{d}}$ is given by $\mathbf{M}_{\text{d}} = \mathbf{A}_{\text{d}}^{1/2}$. 

According to this alternative formulation \eqref{eq:phis}, the points of the decoding lattices are given by $\mathbf{M}_{\text{d}}\B{l}_i$, and their corresponding $\phi_i$ will increase as the Euclidean distance with respect to $\mathbf{M}_{\text{d}}\B{l}_0$ decreases. Hence, we can build the set of candidate vectors as
\begin{equation}
    \mathcal{L}_d = \{\B{l}_i\in \mathcal{Z}^{Kn} ~|~ || \mathbf{M}_{\text{d}}\B{l}_0 - \mathbf{M}_{\text{d}}\B{l}_i)||^2 < R^2\}.  
\end{equation}
This idea resembles the so-called integer least-square problem where the sphere decoder has been shown to be an effective solution \cite{1194444}. In this case, the application of the sphere decoder to construct the set $\mathcal{L}_d$ is the same as in \cite{Suarez17} but considering the particular lattice structure of $\varLambda_{\text{d}}$, and using the Gram matrix $\mathbf{A}_{\text{d}}$ and the sphere center $\B{l}_0$.

After obtaining the set of candidate vectors $\B{l}_i$ corresponding to those truncated Gaussian regions with a significant weight, the \ac{MMSE} estimates are determined as follows
\begin{equation}\label{suma_gauss2}
    \hat{\mathbf{s}}_{\text{c}}=\frac{\sum\limits_{i=1}^{|\mathcal{L}_d|}\phi_i{\Theta}( \Omega_{\varLambda} (\B{l}_{i});\boldsymbol{\Sigma},\boldsymbol{\mu}_i)}{\sum\limits_{i=1}^{|\mathcal{L}_d|}\phi_i
    {\Phi}(\Omega_{\varLambda}(\B{l}_{i});\boldsymbol{\Sigma},\boldsymbol{\mu}_i)}.
\end{equation}
However, this expression still requires the computation of $2|\mathcal{L}_d|$ integrals of $Kn$-dimensional Gaussian functions over complex truncated regions. To circumvent this problem, we propose to approximate the \ac{MMSE} integrals in \eqref{suma_gauss2} by the corresponding \ac{MAP} estimates for each region of the candidate vectors in $\mathcal{L}_d$. In this particular case, the \ac{MAP} and \ac{MMSE} estimators are not strictly equivalent due to the truncated nature of the conditional posterior probability. However, for an adequate design of the mapping parameters, the \ac{MAP} estimates will be an accurate approximation since the peak values of the truncated Gaussian functions will mostly fall into the corresponding truncated regions. Hence, we can simplify the expression in \eqref{suma_gauss2} as
$
    \hat{\mathbf{s}}_{\text{c}}=\sum_{i=1}^{\mid\mathcal{L}_d\mid} \phi_i  \;\hat{\mathbf{s}}^{\text{MAP}}_{i},
$
 where $\hat{\mathbf{s}}^{\text{MAP}}_{i}$ is the \ac{MAP} estimation for the region $\Omega_{\varLambda_{\text{}}} (\B{l}_{i})$,
which is the solution to the following maximization problem
\begin{align}\label{op1}
\hat{\mathbf{s}}^{\text{MAP}}_{i}=\underset{\tilde{\mathbf{s}}_i}{\text{arg}\;\text{max}}\;\;p(\tilde{\mathbf{s}}_{i}\mid \tilde{\mathbf{y}}_{\text{c}})=\underset{\tilde{\mathbf{s}}_{i}}{\text{arg}\;\text{max}}\;\;\frac{p(\tilde{\mathbf{s}}_{i}\mid \tilde{\mathbf{y}}_{\text{c}})~{p(\tilde{\mathbf{s}}_{i}})}{p(\tilde{\mathbf{y}}_{\text{c}})}, 
\end{align}
where the a priori probability $p(\tilde{\mathbf{s}}_{i})$ is given by \eqref{eq:pdf_source} with a covariance matrix $\B{C}_{\tilde{\text{s}}_c} = \B{C}_{\tilde{\text{s}}} \otimes \mathbf{I}_{\frac{n}{2}}$, the conditional probability is given by
\begin{equation}
    p(\tilde{\mathbf{s}}_{i}\mid \tilde{\mathbf{y}}_{\text{c}}) = 
    \frac{1}{\left(\pi\sigma_{\text{n}}^2\right)^{nK}}
    \exp\left(
           -\frac1{\sigma_{\text{n}}^2}  \Vert \tilde{\B{y}_{\text{c}}} - \tilde{\B{H}}_{\text{c}}\mathbf{D}(\tilde{\mathbf{s}}_{i}-\mathbf{B}\mathbf{l}_i) \Vert^2
    \right),
\end{equation}
and the term $p(\tilde{\mathbf{y}}_{\text{c}})$ can be disregarded as it does not depend on $\tilde{\mathbf{s}}_{i}$. This maximization problem can be reformulated as the minimization of the arguments of the two exponential functions in \eqref{op1} which correspond to the a priori and the conditional probabilities, respectively. Therefore, the \ac{MAP} estimates are determined by solving the following optimization problem
\begin{align}\label{op2}
{\hat{\mathbf{s}}^{\text{MAP}}_{i}}= \; &
\underset{\tilde{\mathbf{s}}_{i}}{\text{arg}\;\text{min}}\;\;\Vert\tilde{\mathbf{y}}_{\text{c}}-\tilde{\mathbf{H}}_{\text{c}}\mathbf{D}( \tilde{\mathbf{s}}_{i}-\mathbf{B}\mathbf{l}_i)\Vert^{2}+\frac{\sigma_{\text{n}}^2}{2}\tilde{\mathbf{s}}_{i}^{T}\mathbf{C}_{\tilde{\text{s}}_{\text{c}}}^{-1}\tilde{\mathbf{s}}_{i} \notag\\ 
& \text{s.t.}~~~\Vert \tilde{\mathbf{s}}_{{i,k}}-\mathbf{B}\mathbf{l}_{i,k}\Vert^{2} \leq p_{\text{r}}(\alpha_k),~\forall k,
\end{align}
where the $K$ constraints in \eqref{op2} are imposed to ensure that the $i$-th \ac{MAP} solution falls into the corresponding truncated region $\Omega_{\varLambda} (\B{l}_{i})$. It is important to remark that these constraints aim at approximating the corresponding actual Voronoi regions which cannot be defined analytically for an arbitrary dimension.
As observed, the Euclidean distance between the solution vector for each sensor $\tilde{\mathbf{s}}_{{i,k}}$ and the corresponding centroid $\mathbf{B}\mathbf{l}_{i,k}$ must be lower than the covering radius $p_{\text{r}}(\alpha_k)$ of the scaled version of the encoding lattice at sensor $k$.

The problem in \eqref{op2} can be rewritten in a quadratic form as 
\begin{align}\label{op3}
& {\hat{\mathbf{s}}^{\text{MAP}}_{i}} =\;\;
\underset{\tilde{\mathbf{s}}_{i}}{\text{arg}\;\text{min}}\;\frac{1}{2}\tilde{\mathbf{s}}_{i}^T\mathbf{Q}\tilde{\mathbf{s}}_{i}-\mathbf{v}_i^T\tilde{\mathbf{s}}_{i}\\ \notag
& \text{s.t.}\;\;~\Vert \tilde{\mathbf{s}}_{{i,k}}-\mathbf{B}\mathbf{l}_{i,k}\Vert^{2} \leq p_{\text{r}}(\alpha_k),~\forall k,
\end{align}
where $$\mathbf{Q}=2\mathbf{D}^T\tilde{\mathbf{H}}_{\text{c}}^T\tilde{\mathbf{H}}_{\text{c}}\mathbf{D}+\sigma_{\text{n}}^2\mathbf{C}_{\tilde{\text{s}}_{\text{c}}}^{-1}$$ and $$\mathbf{v}_i=2\mathbf{D}^T\tilde{\mathbf{H}}_{\text{c}}^T(\tilde{\mathbf{y}}_{\text{c}}+\tilde{\mathbf{H}}_{\text{c}}\mathbf{D}\mathbf{B}\mathbf{l}_i).$$ This problem is a variant of a quadratically constrained quadratic program (QCQP) which can be solved efficiently by convex optimization techniques.

\subsubsection{Choice of the Sphere Decoder Radius} In sphere decoding,
there is a trade-off between decoding complexity and estimation accuracy that can be adjusted by means of the sphere radius $R$. If $R$ is too large, there will be too many candidates inside the search hypersphere which leads to an intractable complexity. However, if $R$ is too small, there will be no points inside the sphere. A reasonable guess for $R$ is the covering radius of the lattice which constitutes the smallest radius of the spheres centered at the lattice points that cover the entire space (without holes). This approach guarantees the existence of at least one point inside the sphere \cite{1468474}. However, determining the covering radius for a given lattice is itself hard. Therefore, we need to use an alternative strategy to optimize the value of $R$.

Let $\B{l}^*$ denote the true vector used to encode the source symbols $\tilde{\B{s}}_c$, i.e., $d_l = ||\mathbf{M}_{\text{d}}\B{l}_0 - \mathbf{M}_{\text{d}}\B{l}^* ||^2$ follows a chi-square distribution $\mathcal{X}^2$ with $Kn$ degrees of freedom \cite{1468474}, i.e., $d_l \sim \mathcal{X}^2_{Kn}$. Using this result, we can ensure that the optimum lattice point will fall inside the hypersphere with center $\B{l}_0$ and radius $R$ with a probability $1-\epsilon$ as long as
\begin{equation}
    R^2 \geq F^{-1}_{\mathcal{X}^2_{Kn}}(1-\epsilon),
    \label{eq:radius_sdecoder}
\end{equation}
where $F_{\mathcal{X}^2_{Kn}}(\cdot)$ represents the cumulative distribution function of a chi-square variable with $Kn$ degrees of freedom. Therefore, the $\epsilon$ parameter should be set to a value close to zero to guarantee that the optimum vector is obtained by the sphere decoder with high probability. We have checked experimentally that the criterion in  \eqref{eq:radius_sdecoder} provides a good trade-off for $\epsilon \approx 10^{-5}$. 

In any case, the choice of the radius $R$ is not critical in terms of the system performance, since if no candidates are found for a given $\epsilon$, $R$ can be increased and the sphere decoder is applied again with the new value. Conversely, the value of $R$ does impact the computational cost of the decoding phase. It is thus important to prevent the use of excessively large $R$ values in the decoding operation. 
\subsection{Parameter Optimization}\label{Sub4}
The optimization of the mapping parameters $\{\alpha_k,\delta_k\}$ is fundamental to achieve good performance. Recall that $\alpha_k$ determines the distance between the lattice points at the $k$-th sensor, whereas $\delta_k$ parameters correspond to the power factors employed to satisfy the transmit power constraints. Reducing $\alpha_k$ decreases the size of the associated Voronoi regions and thus the norm of the encoded vectors as they are given by the difference vector between the source vectors and their centroids. 
This in turn impacts on the power factors $\delta_k$ which can be determined as $\delta_k = \sqrt{{P_{\text{T}}}_k / e_k(\alpha_k)}$, where $e_k(\alpha_k)$ is the resulting quantization error for the $k$-th sensor. Note that $e_k(\alpha_k)$ will decrease when lowering $\alpha_k$, which will allow to use larger $\delta_k$ values for a given power constraint ${P_{\text{T}}}_k$. In this way, the optimization procedure will focus on selecting an adequate value of the scaling factors ($\alpha_k$) at each transmitter since the corresponding power factors ($\delta_k$) can be computed subsequently from $\alpha_k$ and the available transmit power $P_{\text{T}k}$.

On the other hand, it is worth remarking that large values of $\delta_k$ reduce the symbol distortion at reception since the error covariance matrix $\boldsymbol{\Sigma}$ in \eqref{pao} inversely depends on $\mathbf{D}=\text{diag}(\delta_1,\ldots,\delta_K)\otimes\mathbf{I}_{n}$. This is clear in the ideal situation where one unique candidate $\B{l}_i$ has a significant weight $\phi_i$ and no decoding ambiguities occur. However, when using too small $\alpha_k$ values, the fading \ac{MAC} and the noise will cause decoding ambiguities which will severely degrade the system performance. 

These ideas can be summarized in the following two points:
\begin{itemize}
    \item Decreasing the value of $\alpha_k$ implies that the lattice points will be closer to each other, the Voronoi regions will be ``smaller", and hence the quantization error will also be smaller. As a result of all this, the norm of the encoded vectors will be lower and the resulting power factor $\delta_k$ will be larger. Thus, we can scale the encoded vectors by a large factor for the same transmit power. According to equation \eqref{pao}, the estimation error will hence be smaller as the covariance matrix $\boldsymbol{\Sigma}$ inversely depends on the power factors $\delta_k$. This statement is true as long as the right vector candidate $\mathbf{l}^*$ is chosen in the decoding procedure with the sphere decoder. 
    \item Decreasing the value of $\alpha_k$ implies that the lattice points will be closer to each other, and hence either the channel distortions or the noise can move the source symbols away from their corresponding centroids or even to other Voronoi regions. The smaller the distance between the lattice points, the greater the probability of this situation happening. This can have a negative impact on the decoding operation, either creating ``ambiguities" (two or more candidate vectors with similar probability) or causing decoding to fail. 
\end{itemize}

Therefore, when optimizing $\alpha_k$, an adequate trade-off is essential to minimize the system distortion. The intuitive idea is to use the minimum possible $\alpha_k$ values which minimize the probability of decoding ambiguities. Recall that the decoding ambiguities are caused by the presence of several candidate vectors $\B{l}_i$ with relevant and similar weights. In the alternative lattice-based formulation, this implies that there are several points in the decoding lattices with a similar distance to the center point given by $\B{l}_0$. Therefore, an adequate criterion for the optimization of $\alpha_k$ is to guarantee that the separation among the points in the decoding lattice ${\varLambda}_d$ is larger than a certain threshold $S$, i.e.,
\begin{equation}
    || \mathbf{M}_{\text{d}}\B{l}_i - \mathbf{M}_{\text{d}}\B{l}_j||^2 \geq S, ~~~ \forall ~\B{l}_i \neq \B{l}_j.
    \label{eq:opt_criterion}
\end{equation}
Note that the lattice expression in \eqref{latr} depends on $\B{B}$ which includes the diagonal matrix $\mathbf{U}=\text{diag}(\alpha_1,\ldots,\alpha_K)$. This way, we can formulate an iterative procedure similar to \cite{Suarez17} which alternatively updates the values of $\alpha_k$ and $\delta_k$ until the criterion in \eqref{eq:opt_criterion} is satisfied. The proposed optimization procedure requires the knowledge of the different channel responses since we need such information to construct the generator matrix $\mathbf{M}_{\text{d}}$ for the decoding lattice.

Finally, an appropriate value for the threshold $S$ should be selected. A conservative value is $S=2R^2$ to ensure that the probability of finding other lattice points at a distance equal or smaller than to the optimum one will be negligible. As shown in the previous subsection, the probability that the distance between the optimum vector $\mathbf{l}^*$ and the center point $\mathbf{l}_0$ is greater than $R^2$, i.e., $d_l > R^2$, is below $\epsilon$ according to \eqref{eq:radius_sdecoder}. By selecting the sphere radius with small $\epsilon$ values (e.g. $\epsilon = 10^{-5}$), the distance between the lattice point corresponding to the optimum vector $\mathbf{l}^*$ and the point corresponding to $\mathbf{l}_0$ in the decoding lattice will be smaller than $R^2$ with very high probability. This implies that the distance from $\mathbf{l}_0$ to any other lattice point will be greater than $R^2$ as the distance between two lattice points is at least $2R^2$. Therefore, the probability of decoding ambiguities vanishes when selecting $S=2R^2$. However, we have experimentally observed that this threshold can be reduced even to $S \approx R^2$ without causing detrimental decoding ambiguities.         
\section{Simulation Results}\label{SR}
In this section, we present the results of computer simulations carried out to evaluate the performance of the proposed lattice-based analog \ac{JSCC} system in a $K\times N_\text{r}$ \ac{SIMO} \ac{MAC} scenario. As mentioned, our focus is on \ac{WSN} scenarios with low latency requirements, and hence we will consider significantly small block sizes compared to traditional digital systems.   

At each time instant, the vector of $K$ source symbols is generated from a zero-mean multivariate circularly symmetric Gaussian distribution with covariance matrix  $\B{C}_{\B{s}}$. We assume a correlation model where the source symbols are normalized and the cross-correlation between any two symbols is the same, i.e.,  $[\B{C}_{\B{s}}]_{i,i} = 1 ~\forall i = 1, \ldots, K$ and $[\B{C}_{\B{s}}]_{i,j} = \rho ~\forall i \neq j$. This assumption is adopted for simplicity. However, we highlight that the correlation exploitation in the decoding stage is independent of the correlation model. As mentioned, blocks of $n$ source symbols are encoded at each sensor via an analog \ac{JSCC} lattice-based mapping where the parameters are properly optimized as in \Cref{Sub4} to avoid decoding ambiguities. %
The resulting encoded symbols are sent to the central node over a fading \ac{SIMO} \ac{MAC}, where the channel coefficients 
follow a Rayleigh distribution. The fading channel response is assumed to remain constant during the transmission of $B_{\text{s}}$ blocks of $n$ source symbols.  
At the central node, the vector $\tilde{\mathbf{y}}_{\text{c}}$, with all the received symbols corresponding to a block, is employed to estimate the source symbols of all sensors with the help of the sphere decoder. This simulation procedure is repeated for $C_{\text{R}}$ different channel realizations. %

The system performance is evaluated in terms of the \ac{SDR} obtained for a given range of \ac{SNR} values. The \ac{SDR} is defined as $\text{SDR (dB)} = 10\log_{10}(1/\hat{\xi})$,
where
\begin{align}
\hat{\xi} = \frac1{C_{\text{R}}B_{\text{s}}nK}\sum_{l=1}^{C_{\text{R}}} \sum_{j=1}^{B_{\text{s}}n} \sum_{k=1}^K |s_{k,j,l} - \hat{s}_{k,j,l}|^2
\end{align}
is the average MSE between the source and the estimated symbols. Hence, the SDR is a suitable metric to illustrate the reliability in the transmission of the information measured/acquired by the sensors.
For simplicity, we assume that the available power at the $K$ sensors is the same, i.e., ${P_{\text{T}}}_k = {P_{\text{T}}}, ~\forall k$, whereas the noise component is $\sigma_{\text{n}}^2 =1$. Therefore, the system \ac{SNR} is $\text{SNR}\;\text{(dB)} = 10\log_{10}({P_{\text{T}}})$. 
\subsection{Performance Evaluation of the Lattice-Based Analog \ac{JSCC}} The following lattices have been considered in the simulation experiments carried out:
\begin{enumerate}

  \item Craig's lattices with dimensions $n \in\{ 16, 36, 52\}$ and constructed as in \Cref{sub:alter_const}. The parameter $m$ is chosen to obtain the densest possible lattices for each $n$ with an affordable encoding computational cost.
    For dimensions up to $36$, we consider the densest Craig's lattices, i.e., $m=3$ for $n=16$ and $m=5$ for $n=36$ whereas for $n=52$, $m=3$ is considered since the densest lattice in this high dimension leads to impractical computational complexity in the encoding process. 
	\item The Leech lattice ($n=24$) \cite{cohn2009optimality,borcherds1985leech,conway2013sphere}.
	\item The Barnes-Wall $(\text{BW}_{16})$ lattice ($n=16)$ \cite{conway2013sphere}.
	\item The $E_8$ lattice ($n=8$) \cite{conway2013sphere,5962670}.
\item The bi-dimensional hexagonal lattice ($n=2$) \cite{9097147}.
\item The modulo-like mappings ($n=1$) \cite{5962670,Suarez17}.
	\end{enumerate}

The lattices corresponding to the construction A have been disregarded because of their poor performance as was anticipated by their quantization errors (see \Cref{tableq}). In addition to the above schemes, two performance bounds were considered as benchmarks. One results from ``uncoded" transmission which provides the best performance achievable assuming a zero-delay linear strategy. In uncoded transmission, each  source symbol is multiplied by a complex-valued scalar to exploit both the channel information and the spatial correlation while satisfying the individual power constraints (c.f. \cite{Suarez18b}). The other is the \ac{OPTA}, which corresponds to the best performance achievable by any communication system designed according to the separation principle. The \ac{OPTA} can be determined by equating the source rate-distortion region and the capacity region of the \ac{MAC} \cite[Appendix C]{Suarez17}. Finally, we have also considered a non-lattice JSCC scheme based on channel optimized vector quantization (COVQ) \cite{farvardin1991performance} to compare the performance of the proposed JSCC lattice-based system with another state-of-the-art encoding strategy. In particular, we have adapted the Linde-Buzo-Gray (LBG) algorithm \cite{linde1980algorithm} to produce an optimized channel codebook for the considered WSN SIMO MAC scenario.

\begin{figure}[t!]
	\centering
	\vspace{-2mm}
	\includegraphics[width=0.93\columnwidth]{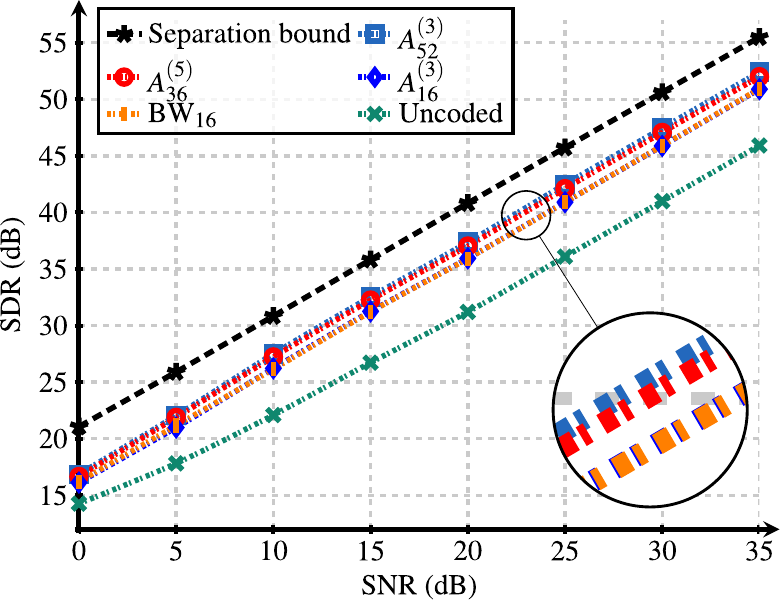}
\vspace{0mm}	\caption{SDR (dB) for different sizes of a Craig's lattice-based mapping ($ n \in \lbrace 16, 36, 52  \rbrace$) with $m=3$  for $n=16$, $m=5$  for $n=36$ and $m=3$ for $n=52$, respectively, and for the $(\text{BW}_{16})$ lattice in a $4\times20$ \ac{WSN} \ac{SIMO} \ac{MAC} setup with $\rho=0.95$.}
	\label{fig1}
	\vspace{0mm}
\end{figure}

\Cref{fig1} plots the \ac{SDR} obtained in a $4\times 20$ \ac{SIMO} \ac{MAC} with correlated sources ($\rho=0.95$). This first experiment was set up to analyze the performance obtained with Craig's lattices of increasing block size. Three Craig's lattices with dimensions $n\in\{{16,36,52}\}$ were considered. According to \Cref{tablerp}, the values of $m$ which lead to the densest Craig’s lattices are $m \in \{3,5,7\}$, respectively. %
Note that we have taken the optimal values of $m$ for $n=16$ and $n=36$ while for $n=52$, $m=3$ was selected because is the largest value leading to a reasonable encoding cost for the considered system. We also remark that the proposed alternative construction makes it affordable to perform Craig's lattice encoding for larger values of $m$ when $n=36$, and for some values of $m$ when $n=52$.         

As observed in \Cref{fig1}, the best performance is obtained by the Craig's lattice with the largest dimension $n=52$. This is a very interesting result as it shows that the transmission reliability is improved when increasing the codeword size in spite of not using the best packing lattices. \Cref{fig1} also shows that Craig's lattices are a good choice for encoding since the lattice $A_{16}^{(3)}$ provides the same performance as $\text{BW}_{16}$ which is the densest lattice for $n=16$ \cite{conway2013sphere}. Finally, Craig's lattices allow to reduce the gap of the linear approaches (uncoded transmission) w.r.t. the \ac{OPTA} from 10 dB to only 2 or 3 dB. Nevertheless, note that \ac{OPTA} is actually an optimistic upper bound since infinite block length is assumed for the source and channel encoders, and the constraints for the individual rates are disregarded. %

\begin{figure}[t!]
	\centering
	\vspace{-2mm}
	\includegraphics[width=0.93\columnwidth]{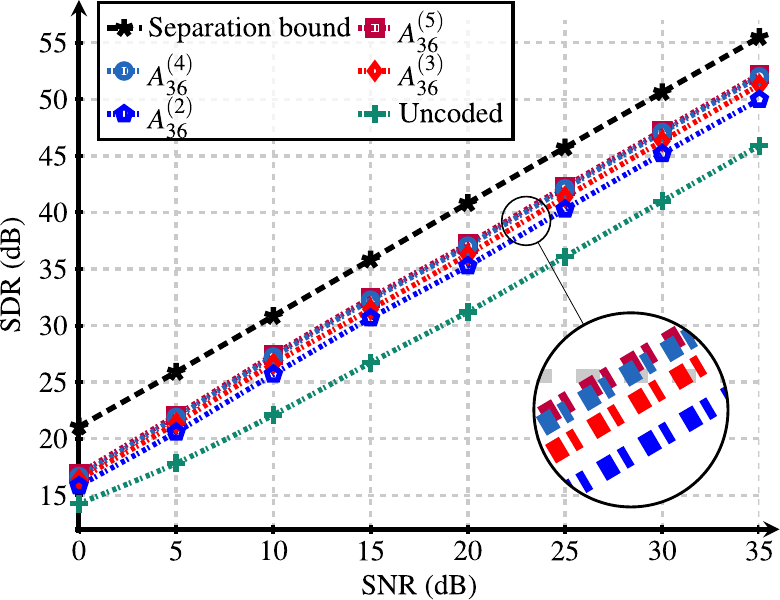}
\vspace{0mm}		\caption{SDR (dB) obtained with  $m\in \lbrace2,3,4,5\rbrace$ by considering a Craig's lattice-based mapping (with size $n=36$) in a $4\times 20$ \ac{WSN} \ac{SIMO} \ac{MAC} setup with spatial correlation $\rho=0.95$.}
\vspace{0mm}
	\label{fig2}
\end{figure}

Next, we evaluated the impact of the lattice density on the system performance. \Cref{fig2} shows the \acp{SDR} achieved in the same communication scenario as before when considering Craig's lattice-based mappings for $n=36$ and $m\in \lbrace2,3,4,5\rbrace$. %
As observed, the \ac{SDR} improves with the lattice density. Indeed, the best performance is achieved for $m=5$, i.e., the densest ${A}_{36}^{(5)}$ lattice provides the highest SDR. This result illustrates the importance of optimizing the lattice density for a given codeword size $n$. In any case, the gain obtained when moving from $m=4$ to $m=5$ is minimum because the increase of the lattice density is also relatively small (see \Cref{tablerp}). We have also observed that the system performance starts to decrease for $m$ values above the optimal one $m_0$.

\begin{figure}[t!]
	\centering
	\vspace{0mm}
	\includegraphics[width=0.93\columnwidth]{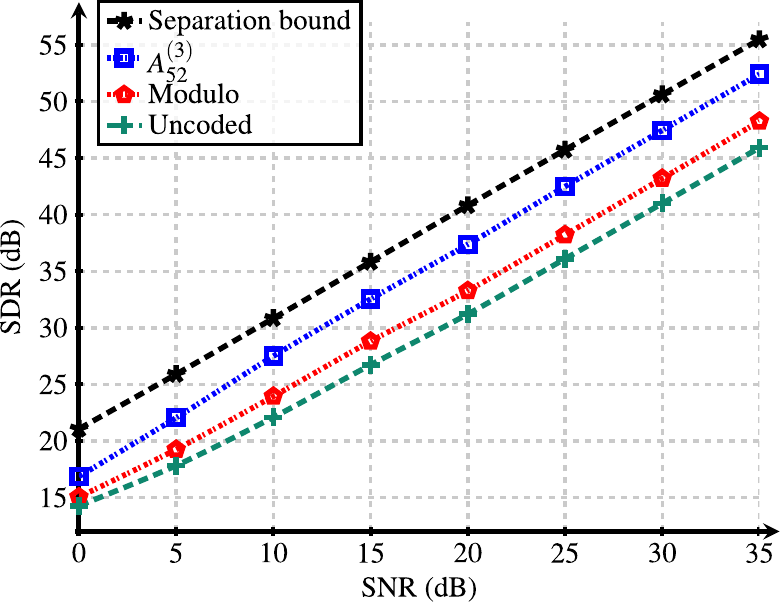}
\vspace{0mm}	\caption{SDR (dB) obtained with different analog lattice-based mappings in a $4\times 20$ \ac{WSN} \ac{SIMO} \ac{MAC} setup with correlation $\rho=0.95$.}
	\label{fig3}
	\vspace{0mm}
\end{figure}

 \begin{table*}
	\centering
	\caption{$\text{SDR}\;(\text{dB})$ obtained with the different \ac{JSCC} mappings for different \acp{SNR} and  $\rho=0.95$.}
	\label{table1}
	\setlength{\tabcolsep}{5.5pt}
	\def\arraystretch{1.6}
\begin{tabular}{ccccccccccc}
		\hline
		$\mathbf{{SNR}\;(\text{dB})}$&
	    \textbf{Sep. Bound}  &
		$\mathbf{{A}_{52}^{(3)}}$ &
		$\mathbf{{A}_{36}^{(4)}}$ &
		\textbf{Leech}  & 
		$\mathbf{{BW}_{16}}$ &  
		$\mathbf{{E}_{8}}$  &
		\textbf{Hexagonal} &
	\textbf{	Modulo}  &
	\textbf	{Uncoded } &
	\textbf	{$\textbf{COVQ}_8$}\\

				0&
	     21.07 &
	16.86	 &
	16.66	 &
	16.39	  & 
	16.23	 &  
	 16.06 &
		15.55 &
		 15.05 &
		 14.26 & 16.77\\

		 	5&
	     25.90 &
	22.05	 &
	21.88	 &
	21.52	  & 
	21.17	 &  
	 20.76 &
		19.97 &
		 19.23 &
		 17.82  & 21.87\\

		 10&
	     30.84 &
	27.50	 &
	27.23	 &
	26.64	  & 
	26.15	 &  
	 25.77 &
		24.71 &
		 23.95 &
		 22.09 & 26.70\\
		 
		 15&
	     35.82 &
32.54	 &
	32.25	 &
	31.61	  & 
	31.21	 &  
	 30.76 &
		29.68 &
		 28.81 &
		 26.72 & 31.52 \\

		  20&
	     40.82&
37.37	 &
	37.03	 &
	36.33	  & 
	35.92	 &  
	 35.38 &
		34.24 &
		 33.28 &
		 31.21 & 36.35 \\

		 25&
	     45.73&
42.49	 &
	42.10	 &
	41.31	  & 
	40.89	 &  
	 40.32 &
		39.11 &
		 38.21 &
		 36.10 & 41.17 \\
		 
		 30&
	     50.62&
47.46	 &
	47.07	 &
	46.21	  & 
	45.90	 &  
	 45.24 &
		44.12 &
		 43.21 &
		 40.99 & 45.91\\

		 	 35&
	     55.48&
52.42	 &
	52.00	 &
	51.25	  & 
	50.93	 &  
	 50.18 &
		49.06 &
		 48.25 &
		 45.91 & 50.69 \\
		\hline

	\end{tabular}
	\vspace{0mm}
\end{table*}

\Cref{table1} shows the \ac{SDR} values obtained with the proposed lattice-based analog \ac{JSCC} approach when using different block sizes (and delays) for the encoding of correlated sources with $\rho=0.95$ in a fading $4\times 20$ \ac{SIMO} \ac{MAC} system. We consider the best packing lattices for each dimension until $n=24$ and two different Craig's lattices for $n=52$ and $n=36$, namely ${A}_{52}^{(3)}$ and ${A}_{36}^{(4)}$. As seen in the previous experiments, these Craig's lattices provide an appropriate balance between performance and computational cost.  We have also included the results obtained for the OPTA bound, the linear system based on uncoded transmission and the non-lattice JSCC scheme based on COVQ with the LBG algorithm. This latter scheme was implemented by assuming dimension $n=8$ as larger block sizes lead to a prohibitive computational complexity because of the huge number of centroids required for a fair comparison with the proposed scheme.  
In addition, \Cref{fig3} illustrates the system performance for the ${A}_{52}^{(3)}$ mapping, the modulo-like mapping, the uncoded transmission and the \ac{OPTA} bound. In the figure, we can appreciate an \ac{SDR} gain of about 2 dB by using the optimized modulo-like mappings instead of the uncoded transmission. This gain is due to the non-linearity of the modulo functions which makes them more suitable for the zero-delay transmission of correlated sources \cite{Wernersson09}. The improvement w.r.t. the modulo-like mappings when using the Craig's lattice ${A}_{52}^{(3)}$ is significantly larger (around $5$ dB at high \ac{SNR} values) which is due to the utilization of a suitable lattice with larger block sizes. Therefore, the proposed lattice-based system is able to improve the reliability of transmission while preserving the delay at a low level.
It is also remarkable that lattice-based analog encoding is able to significantly reduce the gap from the separation bound by assuming practical block sizes in scenarios with low latency requirements, which are significantly smaller than those normally used for digital encoding. Note that the results in \Cref{table1} support these conclusions for all the considered \acp{SNR} and block sizes. In this sense, it is also worth highlighting that COVQ-based scheme provides slightly better performance that its counterpart $E_8$ lattice. However, the resulting SDR values are below those obtained with the Craig's lattices despite the higher computational cost required to optimize the centroid's distribution for each channel realization.

In order to complete this analysis and evaluate the behaviour of the proposed lattice-based scheme in more practical situations we consider the geometric-based stochastic channel model COST2100 \cite{cost2100}. The channel model parameters have been selected to represent an illustrative transmission of information in a WSN. In particular, SIMO channel realizations are generated according to an scenario ``IndoorHall\_5Ghz" with NLoS, carrier frequency of $5.3$ GHz and channel bandwidth of $20$ MHz. We also consider a $10\times10$ m square room where the 4 transmit nodes are placed at the corners of the room and the central node is located at the center. The central receiver is assumed to be equipped with 20 antennas with a half wavelength separation. Finally, the transmit power is properly adjusted at each node to ensure a certain average SNR at the receiver.  

\Cref{figCOST} shows the obtained results for the ${A}_{52}^{(3)}$ Craig mapping, the modulo-like mapping, the uncoded transmission, and the OPTA bound. As observed, the behavior of the different analog JSCC schemes quite resembles the one obtained with Rayleigh channels. The performance gain resulting from the use of lattice-based schemes with larger codewords is similar and the gap with respect to the OPTA bound remains around 3 or 4 dB, although the performance of all considered schemes becomes worse for all the SNR values. This is due to the fact that the channel realizations for the selected COST2100 channel model present larger attenuation and lower spatial diversity than in the Rayleigh case. This effect is particularly visible for low SNR values, where the gain provided by increasing the codeword size is minimum.          

\begin{figure}[t!]
	\centering
	\vspace{0mm}
	\includegraphics[width=0.96\columnwidth]{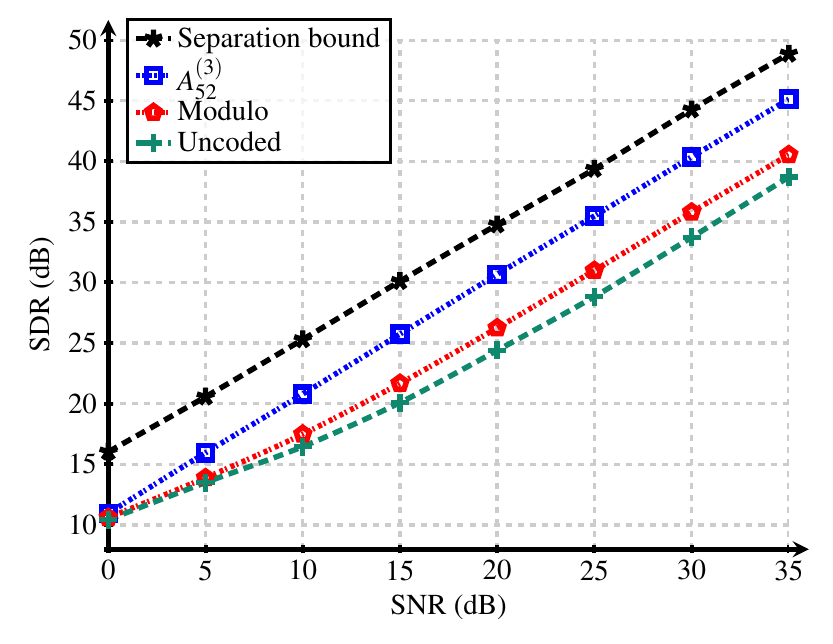}
\vspace{-2mm}	\caption{SDR (dB) for different lattice-based mappings in a $4\times 20$ \ac{SIMO} \ac{MAC} setup with correlation $\rho=0.95$ and assuming the indoor COST2100 channel model.}
	\label{figCOST}
	\vspace{0mm}
\end{figure}

\begin{figure}[t!]
\centering
\vspace{-3.5mm}
\includegraphics[width=0.94\columnwidth]{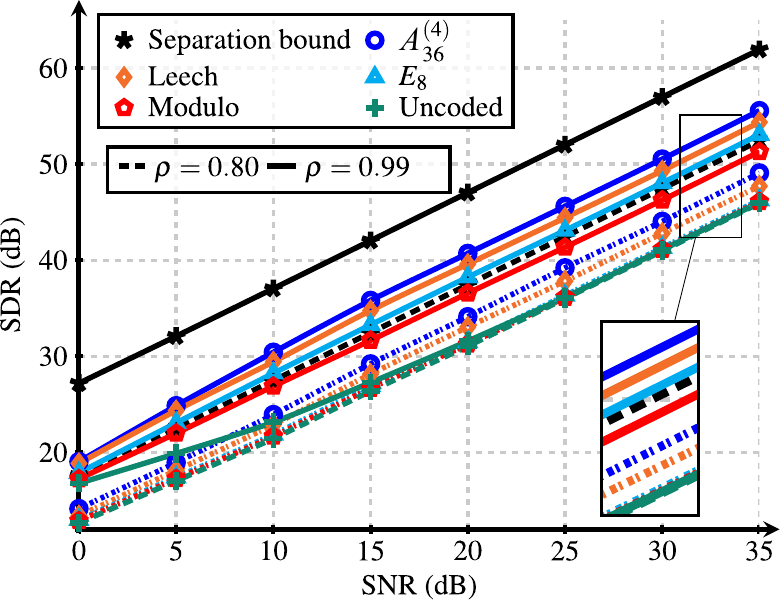}
\vspace{-3mm}
	\caption{SDR (dB) obtained with different analog \ac{JSCC} schemes in a \ac{WSN} \ac{SIMO} \ac{MAC} setup with $N_{\text{r}}=20$, $K=4$ and  $\rho\in \lbrace 0.80, 0.99 \rbrace$.}
	\label{fig4}
	\vspace{0mm}
	\end{figure}
\begin{figure}[ht!]
\begin{minipage}[c]{.49\textwidth}
\centering
\includegraphics[width=0.93\columnwidth]{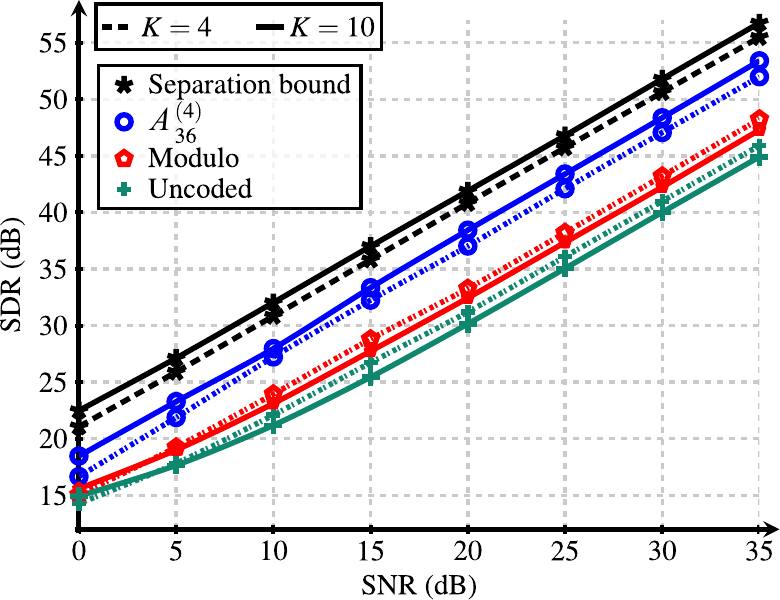}
\vspace{0mm}
	\caption{SDR (dB) obtained with different analog \ac{JSCC} lattice-based mappings in a \ac{SIMO} \ac{MAC} setup with $N_{\text{r}}=20$, $K\in \lbrace 4,10\rbrace$ and $\rho= 0.95 $.}
	\label{fig5}
\end{minipage}
\begin{minipage}[c]{.49\textwidth}
	\centering
	\vspace{4mm}
	\includegraphics[width=0.91\columnwidth]{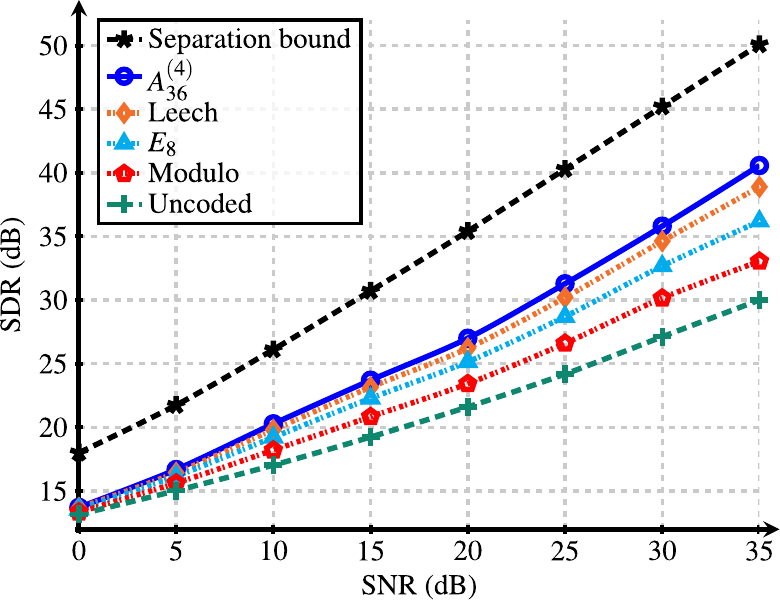}
	\vspace{0mm}
	\caption{SDR (dB) obtained with different analog \ac{JSCC} schemes in a \ac{WSN} \ac{SIMO} \ac{MAC} setup with $N_{\text{r}}=10$, $K=10$ and  $\rho = 0.95 $.}
	\label{fig7}
\end{minipage}
\vspace{0mm}
\end{figure}

We next analyzed performance for different levels of spatial correlation among sensors. \Cref{fig4} plots the \ac{SDR} obtained in a fading $4\times 20$ \ac{SIMO} \ac{MAC} system with two different correlation values, namely $\rho\in \lbrace 0.80, 0.99 \rbrace$, and different analog lattice-based mappings. 
For the lowest correlation level, $\rho=0.8$, modulo-like mappings and $E_8$ lattices provide negligible gain w.r.t. uncoded transmission. In this case, the use of larger encoding blocks is necessary to exploit the spatial correlation among the sources. For example, the lattice $A_{36}^{(4)}$ already achieves an \ac{SDR} improvement of about $4$ dB. This is a remarkable result as one of the major limitations of modulo-like mappings is that they only perform adequately in high correlation scenarios. This behavior changes when the correlation factor becomes larger since the \ac{SDR} gains over uncoded transmission are noticeable even for zero-delay modulo-like mappings. Such gains gradually increase with the block size. On the other hand, the gap of the analog \ac{JSCC} systems w.r.t. the separation bound apparently increases with the sources correlation level, e.g., with ${A}_{36}^{(3)}$, the gap goes from $2$ dB to $4$ dB. This result hence suggests that analog lattice-based \ac{JSCC} would need larger block sizes to efficiently exploit high correlation levels in the source symbols.

The communication scenarios considered in the previous experiments were favorable for zero-delay modulo-based mappings and uncoded transmissions since the receiver had enough degrees of freedom to handle the potential interference caused by the simultaneous transmission of several sensors. 
However, we next analyze if the encoding with larger block sizes can help to mitigate the performance degradation observed for scenarios with higher levels of interference (i.e., less orthogonal). \Cref{fig5} shows the performance obtained for correlated sources with $\rho=0.95$ and two different \ac{WSN} \ac{SIMO} setups: $4\times 20$ and $10\times 20$. %
It is interesting to observe that the system with modulo-like mappings (and uncoded transmission) leads to higher performance for the $4\times 20$ \ac{SIMO} setup than for the $10\times 20$ configuration. Thus, higher levels of interference prevent the system to efficiently exploit the spatial correlation present in the information measured by the sensors, and therefore the performance degrades when there are more sensors in the system. Conversely, the Craig's lattice-based schemes lead to better performance for the $10\times 20$ \ac{SIMO} \ac{MAC} in spite of having fewer degrees of freedom to cancel the sensor interference. This behavior is similar to that of the upper bound which suggests that the separation-based schemes are able to deal properly with the interference while exploiting the higher overall correlation for the $10\times 20$ configuration if the number of sensors increases. This is another relevant result as it allows us to circumvent other of the major limitations of the zero-delay mappings and approximate with small block sizes the behavior of conventional digital separation-based systems.   

In the ensuing experiment, we aimed at providing more insight into the previous issue. We considered an extreme setup with $K= N_\text{r}$, namely a $10 \times 10$ \ac{SIMO} \ac{MAC}. In this case, the ${A}^{(4)}_{36}$ lattice-based mapping, the Leech lattice-based mapping, the $E_8$ construction, and the modulo-like mappings were employed to encode correlated sources with $\rho=0.95$. From the results in \Cref{fig7}, we can derive two important conclusions: 1) the use of non-zero delay mappings provides larger gains w.r.t. zero-delay mappings than in the previous (more orthogonal) configurations, and 2) the gap w.r.t. the \ac{OPTA} is also larger. The first point becomes clear by comparing the performance of modulo-like mappings to that of the ${A}_{36}^{(3)}$ lattice in Figs. \ref{fig5} and \ref{fig7}. The gain of using the Craig's lattices goes from $5$ dB to $8$ dB when considering a $10\times10$ setup instead of a $10\times 20$ one. The second claim is confirmed by comparing the gap between the \ac{SDR} curve for the Craig's lattice and the one for the \ac{OPTA}. This gap goes from $3$ dB to almost $10$ dB when we move to the $10\times10$ setup. This analysis supports that the use of larger block sizes helps to mitigate the impact of high levels of interference. Nevertheless, it would be required to further increase the block size to closely approach the \ac{OPTA} at expense of penalizing the communication delay.

\subsection{Optimization of lattice-based analog \ac{JSCC} system}
In this subsection, we consider some details about the optimization of the proposed lattice-based analog \ac{JSCC} approach and the complexity of the decoding operation.

\subsubsection{Parameter $S$}

As commented in \Cref{Sub4}, the optimization of the mapping parameters is fundamental to achieve adequate system performance. The trade-off between reducing the symbol distortion and avoiding decoding ambiguities when selecting the $\alpha_k$ parameters is managed by the threshold $S$. In \Cref{Sub4}, we provide some insight into an adequate choice of $S$. However, in the following, we experimentally evaluate the accuracy of this choice.

\Cref{figop} shows the \ac{SDR} versus reasonable values for the parameter $S$ when using the Craig's lattice $A_{36}^{(4)}$ and the Leech lattice in a $4\times 20$ \ac{SIMO} \ac{MAC} with $\rho\in\lbrace 0.80,0.95\rbrace$. As observed, the highest system performance is obtained when $S=R^2$ for both lattices, i.e., a proper value for $S$ is to be found in the order of $R^2$. These results hence confirm our initial hypothesis.  
\Cref{figop} also shows that the system performance dramatically degrades when $S < R^2$ because the lattice points are too close to each other leading to inevitable decoding ambiguities. %
This effect is less severe for $S > R^2$, but the resulting system performance is not optimal as we are using too large values for $\alpha_k$ parameters. At the same time, performance seems more sensitive to the adjustment of $S$ for high levels of source correlation. 
\begin{figure}[t!]
	\centering
		\vspace{0mm}
	\includegraphics[width=0.92\columnwidth]{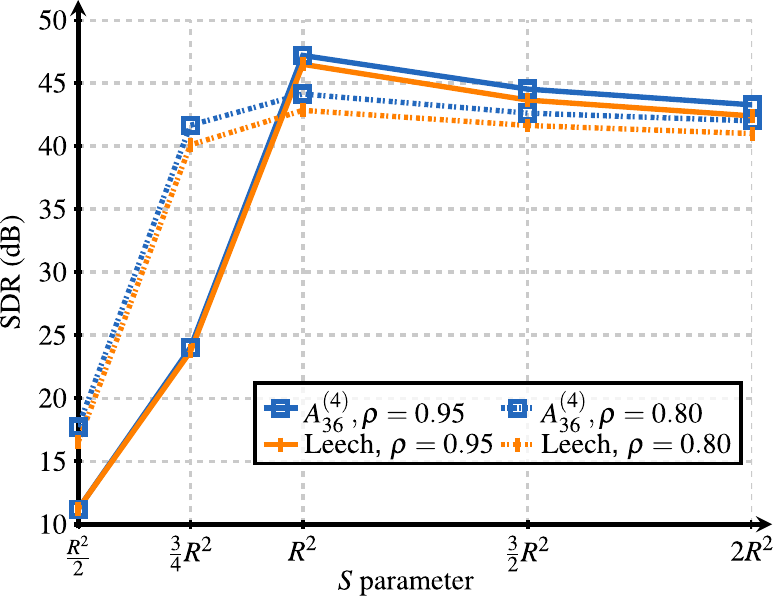}\vspace{0mm}
	\caption{SDR (dB) obtained with different values for the parameter $S$ by encoding with a Craig's lattice-based mapping $({A}^{(4)}_{36})$ and the Leech lattice in a \ac{SIMO} \ac{MAC} setup with $N_{\text{r}}=20$, $K=4$ and  $\rho\in\lbrace 0.80,0.95\rbrace$.}
	\label{figop}
	\vspace{0mm}
\end{figure}

\subsubsection{Parameter $R$} The sphere decoder radius $R$ is set according to \eqref{eq:radius_sdecoder}. The choice of $R$ is less critical than that of $S$ because the sphere decoder can be applied again with a larger $R$ if no candidates are found. However, a proper choice is important to avoid repeating the application of the sphere decoder and to limit the number of candidates falling into the sphere. We have observed from the computer experiments that the criterion in \eqref{eq:radius_sdecoder} is a good choice for $R$. 
\begin{figure}[t!]
	\centering
	\includegraphics[width=0.93\columnwidth]{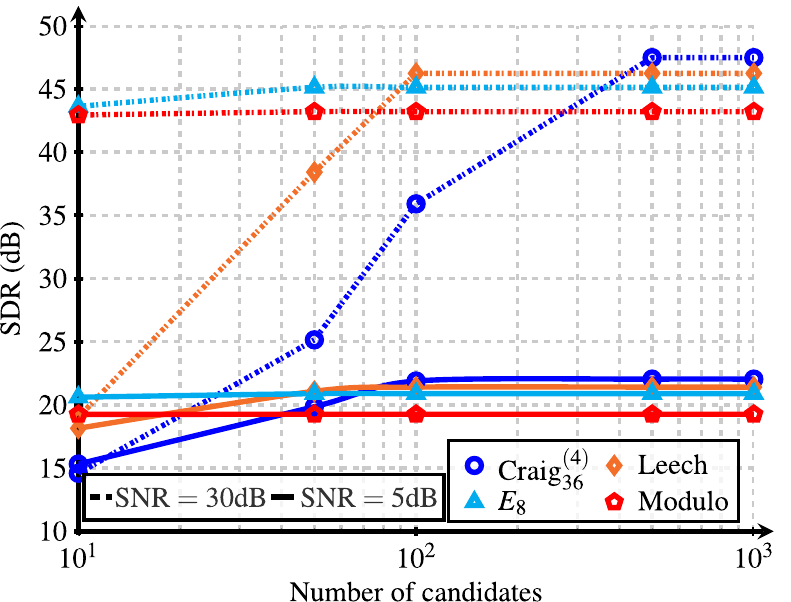}
	\vspace{0mm}
	\caption{SDR (dB) vs maximum number of candidates in the decoding operation with different analog \acp{JSCC} mappings in a  \ac{SIMO} \ac{MAC} setup with $N_{\text{r}}=20$, $K=4$, $\rho= 0.95  $ and  $\text{SNR (dB)}\in \lbrace 5, 30 \rbrace$.}
	\label{fig8}
	\vspace{-2.5mm}
\end{figure}

\subsubsection{Maximum Number of Candidates}
The overall computational complexity of the proposed lattice-based analog \ac{JSCC} system is determined by the encoding and decoding stages. The operations with the highest computational cost are the search for the closest lattice point given the source vector at the encoder and the search for the candidate vectors in $\mathcal{L}_d$ at the decoder.
In this sense, the system was properly designed to get a limited number of candidate vectors with a relevant weight and this, together with the use of \ac{MAP} estimates, reduces the decoding computational cost. However, the iterative nature of the sphere decoder demands the consideration of large numbers of potential candidates in intermediate iterations, especially for high dimensions. In those cases, we should limit the maximum number of candidate vectors at the end of each iteration in the sphere decoder but minimizing the probability of disregarding the optimal vector $\B{l}^*$.   

\Cref{fig8} plots the SDR obtained when varying the maximum number of candidates considered in the decoding phase for different lattice-based mappings in a $4\times 20$ \ac{SIMO} \ac{MAC} with $\rho=0.95$. The system performance obtained when using the Craig's lattice $A_{36}^{(4)}$, the Leech lattice, the $E_8$ construction and the modulo-like mappings are compared for $\text{SNR (dB)}\in \lbrace 5, 30 \rbrace$. As observed, Craig's lattice-based and Leech lattice-based mappings require a higher number of candidates to achieve the best behavior at both SNR levels. This is an expected result because the number of potential combinations of Voronoi regions significantly increases with the lattice dimension $n$.
It is also worth remarking that the impact of excessively limiting the number of considered candidates is less critical in the low \ac{SNR} regime.  
The modulo-like mappings are the simplest ones for decoding since they require the smallest number of candidates to reach their best behavior. Finally, the $E_8$ lattice construction leads to a better performance than that obtained with the modulo-like mapping while exhibiting similar decoding complexity.

\subsection{Analysis of other performance indicators}

In this section, we briefly analyze other relevant parameters to complete the performance evaluation of the proposed lattice-based system. In particular, we will focus on the computational complexity, the required bandwidth, and the communication delay. 

Regarding the computational cost, the lattice-based encoding procedure and the search for the best candidate vectors with the sphere decoder are the two operations with the highest complexity. The remaining operations, including the parameter optimization and the computation of the MAP estimates, clearly have lower computational cost. The approaches to search for the closest centroids in the encoding by using the closest point algorithm and for the candidate vectors in the decoding with the sphere decoder actually follow a similar philosophy. Indeed, both approaches are based on iteratively evaluating the lattice points which fall inside an n-dimensional hypersphere. In this sense, the computational complexity for these searches grows exponentially with the lattice dimension (codeword size) \cite{1019833,1468474}. 

However, different strategies can be applied to mitigate the computational cost corresponding to the standard search on n-dimensional hyperspheres. On the one hand, the maximum number of candidates provided by the sphere decoder at each iteration can be limited depending on the considered SNR value and the lattice type, as shown in the previous subsection. This approach significantly reduces the complexity of the decoding operation as we are restricting the number of points to be evaluated at each sphere decoder iteration. The corresponding computational cost still grows with the lattice dimension since the maximum number of candidates required to provide satisfactory performance must be larger (see \Cref{fig8}), but this approach leads to an acceptable computational complexity at the central node for most applications.   

On the other hand, a similar approach could be considered for the search of the closest lattice point in the encoding procedure. In addition, the use of Craig's lattices provides an additional degree of flexibility to adjust the center density for a given dimension, and hence the number of lattice points to be explored with the closest point algorithm. Moreover, the alternative construction proposed for this type of lattices in this work allows for further lowering the computational cost of the encoding operation. The closest point algorithm is based on recursively representing an $n$-dimensional lattice by $(n-1)$-dimensional parallel translated sub-lattices (layers) \cite{1019833}. In this sense, the dimensionality of the problem can be reduced by ``separating" these layers as this minimizes the number of layers to be explored. It is also desirable that the zero-dimensional layers are as densely spaced as possible \cite{1019833}. %
The use of reduction techniques that guarantee that the scalar product of the the generator matrix columns is as small as possible contributes to satisfy these two conditions  \cite{micciancio2002complexity}.

In summary, the encoding computational complexity is generally in the exponential order with the lattice dimension, but this complexity can be reduced for general lattices by limiting the number of lattice points to be explored. This way, the closest point problem can be solved with an affordable computational cost for the codeword sizes considered in this work. For Craig's lattices, the computational cost of searching the closest lattice point can further be lowered by properly adjusting the center density and constructing a suitable generator matrix.

Another important parameter to assess the system performance is the bandwidth required to transmit the source information. In this case, the required bandwidth is the same regardless of the codeword size or the lattice construction used to encode the source symbols. Assuming that the available channel bandwidth is $B$ Hz, the proposed lattice-based system is designed to transmit $n$ source symbols at $n/2$ channel uses. Therefore, the transmission rate does not depend on the codeword size as it can directly be determined as $v_b=2/T_s$ symbols/s, where $T_s$ corresponds to the channel use duration and  is given by the available channel bandwidth. As a consequence, the spectral efficiency is also equivalent for all the system configurations.   

Finally, the encoding and decoding delays are proportional to the codeword size, or equivalently to the lattice dimension. Note that this delay is minimum compared to the traditional digital systems which use significantly larger codeword sizes. In this way, we can configure the proposed system to achieve different communication reliabilities (SDRs) with minimum impact on the delay, while transmitting the source information at a constant rate.

\section{Conclusions}\label{Con}
This work studied analog \ac{JSCC} lattice-based mappings for the transmission of correlated sources in fading  %
\ac{SIMO} \ac{MAC} systems. These mappings can be particularly suitable for WSNs and IoT systems where it is important to guarantee low communication latency. In these scenarios, the proposed analog \ac{JSCC} scheme allows the use of lattice constructions with different dimensions to encode the sources with variable block sizes. This feature allows for improving the system performance by adjusting the codeword size according to the delay constraints of the application.
At the {central node}, the \ac{MMSE} estimates of the source symbols are computed with the help of {sphere decoding and } \ac{MAP} estimation to reduce the decoding computational complexity. This approach, together with an adequate optimization of the system parameters, enables the transmission of reasonable large block sizes with a good trade-off between performance and computational cost while preserving minimum latency.  
In this sense, Craig's lattices are particularly attractive since they allow for balancing the lattice density to improve transmission reliability with an affordable encoding computational complexity.

Simulation results show that the proposed lattice-based analog \ac{JSCC} system is able to reduce the symbol distortion, and thus increasing the system reliability, as the lattice dimension becomes larger. This improvement is especially remarkable for scenarios with non-orthogonal configurations 
and low levels of spatial correlation. In such scenarios, zero-delay mappings hardly provided some gains with respect to linear approaches. The results also validate the suitability of the proposed Craig's lattices construction for different applications based on low latency \acp{WSN}.
\vspace{0mm}
\bibliographystyle{IEEEtran}
\bibliography{references}

\vfill

\end{document}